\documentclass[11pt]{article}

\title{Distributed Verification and Hardness of Distributed Approximation\thanks{The preliminary version of this paper appeared as \cite{DasSarmaHKKNPPW11} in the Proceeding of the 43rd ACM Symposium on Theory of Computing,
               (STOC) 2011.}}

\author{Atish {Das Sarma}\thanks{Google Research, Google Inc., Mountain View, USA. \hbox{E-mail}:~{\tt dassarma@google.com}.} \thanks{Part of the work done while at Georgia Institute of Technology.} 
\and Stephan Holzer\thanks{Computer Engineering and Networks Laboratory (TIK), ETH Zurich, CH-8092 Zurich, Switzerland. E-mail: {\tt \{stholzer, wattenhofer\}@tik.ee.ethz.ch.}} \and Liah Kor\thanks{Department of Computer Science and Applied Mathematics, The Weizmann
Institute of Science, Rehovot, 76100 Israel.
E-mail: {\tt \{liah.kor,david.peleg\}@weizmann.ac.il}.
Supported by a grant from the United States-Israel Binational Science
Foundation (BSF).}
\and Amos Korman\thanks{CNRS and LIAFA, Univ. Paris 7, Paris, France.
E-mail: {\tt amos.korman@liafa.jussieu.fr}. Supported by the ANR projects ALADDIN and PROSE and by the INRIA project GANG.}
\thanks{Supported by a France-Israel cooperation grant
(``Mutli-Computing'' project) from the France Ministry of Science and Israel Ministry of Science.} \and Danupon Nanongkai\thanks{Faculty of Computer Science, The University of Vienna, Vienna, Austria.
\hbox{E-mail}:~{\tt danupon@gmail.com}}\addtocounter{footnote}{-6}~~\footnotemark \addtocounter{footnote}{5}
\and Gopal Pandurangan\thanks{Division of Mathematical
Sciences, Nanyang Technological University, Singapore 637371 \& Department of Computer Science, Brown University, Providence, RI 02912, USA.  \hbox{E-mail}:~{\tt gopalpandurangan@gmail.com}. Supported in part by the following grants: Nanyang Technological University grant M58110000, US NSF grant CCF-1023166, and a grant from the US-Israeli Binational Science Foundation (BSF).} \and David Peleg\addtocounter{footnote}{-5}\footnotemark\addtocounter{footnote}{4}~~\addtocounter{footnote}{-3}\footnotemark\addtocounter{footnote}{2} \and Roger Wattenhofer\addtocounter{footnote}{-6}\footnotemark
}

\usepackage[latin1]{inputenc}
\usepackage{tikz}
\usetikzlibrary{shapes,arrows,positioning,shadows,snakes}

\usepackage{psfrag}

\usepackage{subfigure}
\usepackage{amsthm,amsmath,amsfonts}
\usepackage{color}
\usepackage{amsfonts}
\usepackage{wrapfig}
\usepackage{fullpage}

\def\figspace{0}
\newcommand{\graph}{G(\Gamma, d, p)}
\newcommand{\ElkinGraph}{\graph}

\def\EQfunc{\mbox{\tt eq}}

\def\DISJfunc{\mbox{\tt disj}}

\def\cP{\mathcal{P}}
\def\cH{\mathcal{H}}
\def\cA{\mathcal{A}}
\def\cT{\mathcal{T}}

\def\indVarS{X^s}
\def\indVarR{X^r}

\newtheorem{theorem}{Theorem}[section]
\newtheorem{corollary}[theorem]{Corollary}

\newtheorem{lemma}[theorem]{Lemma}
\newtheorem{claim}[theorem]{Claim}

\theoremstyle{definition}
\theoremstyle{observation}\newtheorem{observation}[theorem]{Observation}
\newcommand{\squishlist}{\begin{itemize}}
\newcommand{\squishend}{\end{itemize}}

\def\danupon#1{}
\def\stephan#1{}
\def\atish#1{}

\begin{document}

%\begin{titlepage}
\maketitle

\begin{abstract}
We study the {\em verification} problem in distributed  networks, stated as follows. Let $H$ be a subgraph of a network $G$ where each vertex of $G$ knows which edges incident on it are in $H$. We would like to verify whether $H$ has some properties, e.g., if it is a tree or if it is connected (every node knows at the end of the process whether $H$ has the specified property or not). We would like to perform this verification in  a decentralized fashion via a distributed algorithm. The time complexity of verification is measured as the number of rounds of distributed communication.

In this paper we initiate a systematic  study  of  distributed verification, and give almost tight lower bounds on the running time of distributed verification algorithms for many fundamental problems such as connectivity, spanning connected subgraph, and $s-t$ cut verification. We then show applications of these results in deriving strong unconditional time lower bounds on  the {\em hardness of distributed approximation} for many classical optimization problems including minimum spanning tree, shortest paths, and minimum cut. Many of these results are the first non-trivial lower bounds for both exact and approximate distributed computation and they resolve  previous open questions. Moreover, our unconditional lower  bound of approximating minimum spanning tree (MST)  subsumes and  improves  upon  the previous hardness of approximation bound  of Elkin [STOC 2004] as well as the lower bound for  (exact) MST computation  of Peleg and Rubinovich [FOCS 1999]. Our  result implies that there can be no distributed approximation algorithm for MST that is significantly faster than the current exact algorithm, for {\em any} approximation factor.

Our lower bound proofs show an interesting connection between communication complexity and distributed computing which turns out to be  useful in establishing  the time complexity of exact and approximate distributed computation of many  problems.
\end{abstract}

%\begin{keywords}
%distributed algorithms, graph optimization problems, lower bounds, hardness of approximation, communication complexity
%\end{keywords}
%
%\begin{AMS}
%68W15, 05C85, 68M10
%%97P20  	Theory of computer science
%%05C85  	Graph algorithms
%%68Q25  	Analysis of algorithms and problem complexity
%%68W40  	Analysis of algorithms
%%68W15  	Distributed algorithms
%%68W25  	Approximation algorithms
%%68M10  	Network design and communication
%\end{AMS}

%\pagestyle{myheadings}
%\thispagestyle{plain}

%\markboth{Das Sarma, Holzer, Kor, Korman, Nanongkai, Pandurangan, Peleg and Wattenhofer}{Distributed Verification and Hardness of Distributed Approximation}
%\end{titlepage}

\section{Introduction}\label{sec:intro}

Large and complex networks, such as the human society, the Internet, or the brain, are being studied intensely by different branches of science. Each individual node in such a network can directly communicate only with its   neighboring nodes.
Despite being restricted to such \emph{local} communication, the network itself should work towards a \emph{global} goal, i.e., it should organize itself, or deliver a service.

In this work we investigate the possibilities and limitations of distributed/decen-tralized computation, i.e., to what degree local information is sufficient to solve global tasks. Many tasks can be solved entirely via local communication, for instance, how many friends of friends one has.
Research in the last 30 years has shown that  some classic combinatorial optimization problems such as matching, coloring, dominating set, or approximations thereof can be solved   using small (i.e., polylogarithmic)  local communication. For example,
a maximal independent set can be computed in time $O(\log n)$ \cite{Luby86}, but not in time $\Omega(\sqrt{\log n / \log\log n})$ \cite{kuhn-podc04} ($n$ is the  network size). This lower bound even holds if message sizes are unbounded.

However many important optimization problems  are ``global'' problems from
the distributed computation point of view. To count the total number of nodes, to determining the diameter of the system, or to compute a spanning tree, information necessarily must travel to the farthest nodes in a system. If exchanging a  message over a single edge costs one time unit, one needs $\Omega(D)$ time units to compute the result, where $D$ is the network diameter. If message size was unbounded, one can simply collect all the information in $O(D)$ time, and then compute the result. Hence, in order to arrive at a realistic problem, we need to introduce communication limits, i.e., each node can exchange messages with each of its neighbors in each step of a synchronous system, but each message can have at most $B$ bits (typically $B$ is small, say
 $O(\log n)$).  However, to compute a spanning tree, even single-bit messages are enough, as one can simply  breadth-first-search the graph in time $O(D)$ and this is optimal \cite{peleg}.

But, can we \emph{verify} whether an existing subgraph that is claimed to be a spanning tree indeed is a correct spanning tree?!   In this paper we show that this is not generally possible in $O(D)$ time -- instead one needs $\Omega(\sqrt{n} + D)$ time. (Thus, in contrast to traditional non-distributed  complexity, verification is harder  than computation in the distributed world!).  Our paper is more general, as we show interesting lower and upper bounds (these are almost tight) for a whole selection of verification problems. Furthermore, we show
a key application of studying such verification problems to proving
strong  unconditional time lower bounds on exact and approximate distributed computation for many classical problems.

\subsection{Technical Background and Previous Work}\label{subsec:intro:background}
\paragraph{Distributed Computing}
Consider a synchronous  network of processors  with unbounded computational power. The network is modeled by an undirected $n$-vertex graph, where vertices model the processors and  edges model the links between the processors. The processors  (henceforth, vertices) communicate  by exchanging
messages via the links (henceforth, edges).  The vertices  have limited global knowledge, in particular, each of them has its own local perspective of the network (a.k.a graph), which is confined to its immediate neighborhood.  The vertices may have to compute (cooperatively)
some global function of the graph, such as a spanning tree (ST) or a minimum spanning tree (MST), via communicating with each other and running a distributed algorithm designed for the task at hand.
There are several measures to analyze the performance of such algorithms,
a fundamental one being the running time, defined as the worst-case number of {\em rounds} of distributed communication. This measure naturally gives rise to a  complexity measure of problems, called the {\em time complexity}.
On each round at most $B$ bits can be sent through each edge in each direction, where $B$ is the bandwidth parameter of the network. The design of efficient algorithms for this model (henceforth, the $B$ model), as well as establishing lower bounds on the time complexity of various fundamental graph problems, has been the subject of an active area of research called (locality-sensitive)  {\em distributed computing} (see \cite{peleg} and references therein.)

\paragraph{Distributed Algorithms, Approximation, and Hardness}
Much of the initial research focus in the area of distributed computing  was
on designing algorithms for solving problems exactly, e.g., distributed algorithms for ST, MST, and shortest paths are well-known \cite{peleg, lynch}.
Over the last few years, there has been interest in designing distributed algorithms
that provide approximate solutions to problems. This area is known as {\em distributed approximation}. One motivation for designing such algorithms is that they can run faster or have better communication complexity albeit at the cost of providing suboptimal solution. This can be especially appealing for resource-constrained and dynamic networks (such as sensor or peer-to-peer networks). For example, there is not much point in having
an optimal algorithm in a dynamic network if it takes too much time, since
the topology could have changed by that time.
For this reason, in the distributed context,  such algorithms are well-motivated even for network optimization problems that are not NP-hard, e.g., minimum spanning tree, shortest paths etc.
There is a large body of work on distributed approximation algorithms for  various classical graph optimization problems (e.g., see the surveys by Elkin
\cite{Elkin-sigact04} and Dubhashi et al. \cite{Dubhashi}, and the work
of \cite{KhanKMPT08} and the references therein).% \cite{khan-disc, Elkin06}, shortest
%path \cite{Elkin05,KhanKMPT08}, Steiner forest \cite{KhanKMPT08}, minimum cost routing tree \cite{KhanKMPT08}, minimum edge-coloring \cite{Srinivasan}, minimum dominating set \cite{jia02, kuhn-spaa07, dubhashi-dom, kuhn-podc03}, minimum vertex cover \cite{grandoni05},  maximum matching \cite{matching, czygrinow}, facility location \cite{pemmaraju} and positive linear programs \cite{distLP, bartal,distcomb}.

While a lot of progress has been made in the design of distributed approximation
algorithms, the same has not been the case  with the theory of lower bounds on the approximability of distributed problems, i.e., {\em hardness} of distributed approximation.  There are some inapproximability results that are based on lower bounds on the time complexity of the exact solution of certain problems and on integrality of the objective functions of these problems. For example, a fundamental result due to Linial \cite{linial} says that 3-coloring an $n$-vertex ring requires $\Omega(\log^*n)$ time. In particular, it implies that any 3/2-approximation protocol for the vertex-coloring problem requires $\Omega(\log^*n)$ time. On the other hand, one can state inapproximability results assuming that vertices are computationally limited; under this assumption, any NP-hardness inapproximability result immediately implies an analogous result in the distributed model. However,  the above results are not interesting in the distributed setting, as they provide no new insights on the roles of locality and communication \cite{Elkin06}.

There are but a few significant results currently known on the hardness of distributed approximation. Perhaps the first important result was presented for the MST problem by Elkin in
%his STOC 04 paper
\cite{Elkin06}. Specifically, he showed  strong {\em unconditional} lower bounds (i.e., ones that do not depend on complexity-theoretic assumptions)  for distributed approximate MST (more on this result below).
%Elkin's result also generalized the
%seminal
%lower bound result for exact MST due to Peleg and Rubinovich  \cite{PelegR00}.
Later, Kuhn, Moscibroda, and Wattenhofer \cite{kuhn-podc04} showed  lower bounds on time approximation trade-offs for several problems.
%Recently,
%similar results have been shown for other problems such as facility location \cite{pemmaraju}.

  %These results
%employ techniques that are specific to the particular problem at hand.

%While in the centralized setting, there has
%  been remarkable progress in developing uniform techniques for showing
%  hardness of approximation (notably via PCP),  there are no such uniform  approaches to deriving lower bounds for various problems.

\subsection{Distributed Verification}
The above discussion summarized two major research aspects in distributed computing, namely studying distributed algorithms and lower bounds for (1) exact and (2) approximate solutions to various problems.
The third aspect --- that  turns out to have remarkable applications to the first two ---  called {\em distributed verification}, is the main subject of the current paper. In distributed verification, we want to efficiently check
whether a given subgraph of a network has a specified property via a distributed algorithm\footnote{Such problems have been studied  in the sequential setting, e.g., Tarjan~\cite{tarjan} studied
verification of MST.}.
%For example, we
%might want to check if the subgraph is a spanning tree or if it is
%connected
Formally, given a  graph $G=(V,E)$,  a subgraph $H=(V,E')$ with $E'\subseteq E$, and a predicate $\Pi$,
it is required to decide whether $H$ satisfies $\Pi$ (i.e., when the algorithm terminates, every node knows whether $H$ satisfies $\Pi$).  The predicate $\Pi$ may specify statements such as ``$H$ is connected'' or ``$H$ is a spanning tree" or ``$H$ contains a cycle''. (Each vertex in $G$ knows which of its incident edges (if any) belong to $H$.)
The goal is to study bounds on the time complexity of distributed verification. The time complexity
of the verification algorithm is measured with respect to parameters of $G$
(in particular, its size $n$ and diameter
%\footnote{The length of a path $p$ in $G$ is the number of edges it contains. The {\em distance} between two vertices is the length of the shortest path connecting them. The {\em diameter} $D$ of $G$ is the maximum distance between any two vertices of $G$.}
$D$), independently from $H$.

%While we consider
%a large class of  problems in this paper, we note that all of them are unified
%by the fact that they are ``global" problems (which require algorithms to
%traverse the entire network). Network diameter of $G$ will
%be an inherent lower bound on the time complexity.
%As pointed out in the seminal work of \cite{peleg-mst},
%for such problems it would be
%desirable to focus on

We note that verification is different from construction  problems, which have been the traditional focus in distributed computing.
Indeed, distributed algorithms for constructing spanning trees, shortest paths, and other problems  have been well studied (\cite{peleg, lynch}). However, the corresponding verification problems
have received much less attention. To the best of our knowledge,  the only distributed verification problem that has received some attention is the
MST (i.e., verifying if $H$ is a MST); the recent work of Kor et al. \cite{KorKP10} gives a $\Omega(\sqrt{n}/B + D)$ deterministic lower bound on distributed verification of MST, where $D$ is the diameter
of the network $G$.
% (However, this lower bound applies only for deterministic
%algorithms.)
That paper also gives a matching upper bound (see also \cite{amos-kutten}).
Note that distributed {\em construction} of MST has rather similar lower and upper bounds \cite{PelegR00, peleg-mst}. Thus in the case of the MST problem, verification and construction
have the same time complexity.
 We later show that the above result of Kor et al. is subsumed by the results of this paper, as we show that verifying {\em any}
spanning tree takes so much time.

\paragraph{Motivations} The study of distributed verification has two main motivations. The first is understanding the complexity of verification versus construction. This is obviously a central question in the traditional RAM model, but here we want
to focus on the same question in the distributed model. Unlike in the centralized setting, it turns out that verification is {\em not} in general easier than construction in the distributed setting!
In fact, as was indicated earlier, distributively verifying a spanning tree
turns out to be harder than constructing it in the worst case.
Thus understanding the complexity
of verification in the distributed model is also important. Second, from an algorithmic point of view, for some problems, understanding the verification
problem can help in solving the construction problem  or showing the inherent limitations in obtaining an efficient algorithm.  In addition to these, there is
yet another motivation that emerges from this work: We show that distributed verification leads to showing {\em strong unconditional lower bounds on  distributed  computation (both exact and approximate)} for a variety of problems, many hitherto unknown.
For example, we show that establishing a lower bound on the spanning connected subgraph verification problem leads to establishing lower bounds
for the minimum spanning tree, shortest path tree, minimum cut etc.  Hence, studying  verification problems may lead to proving hardness of approximation as well as lower bounds for exact computation
for new problems.

%The focus of this paper will be on such fundamental verification problems on unweighted graphs.\footnote{However, one
%can naturally extend the problem on weighted graphs, in which case  $\Pi$ can be ``$H$ is a MST" or ``$H$ is a shortest path tree".}

% There are many problems considered, we will list them in the corresponding sections.
%
%\paragraph{The spanning tree (ST) verification problem}
%is defined as follows. Given a graph $G(V,E)$, and a subgraph  $T$ of $G$,
%referred as the {\em ST candidate}, it is required to decide whether  $T$
%forms a spanning tree on $G$, i.e., $T$ is a cycle free connected subgraph
%of $G$.
%
%\paragraph{The cycle detection problem} is defined as follows.
%Given a graph $G=(V,E)$ and a subgraph $H=(V,E')$ of $G$, it is required
%to decide whether $H$ is cycle free.
%
%\paragraph{The subgraph connectivity problem} is defined as follows.
%Given a graph $G=(V,E)$ and a subgraph $H=(V,E')$ of $G$, it is required
%to decide whether $H$ is connected.

\subsection{Our Contributions}
In this paper, our main contributions are twofold. First, we initiate a
systematic  study  of {\em distributed verification}, and give almost tight uniform lower bounds on the running time of distributed verification algorithms for many fundamental problems.
  Second, we make progress in establishing strong hardness results on the distributed approximation  of many classical optimization problems. Our lower bounds also apply seamlessly
  to exact algorithms. We next state our main results (the precise theorem statements are in
  the respective sections as mentioned below).

%We present an uniform approach to obtaining (essentially) tight lower bounds for
%many fundamental distributed verification problems. Our approach also leads to a
%uniform approach to showing {\em hardness of distributed approximation} for
%many fundamental problems.

\paragraph{1. Distributed Verification} We show a  lower bound of $\Omega(\sqrt{n/(B\log n)}+D)$ for many verification problems in the $B$ model, including {\em spanning connected subgraph}, {\em $s$-$t$ connectivity}, {\em cycle-containment}, {\em bipartiteness}, {\em cut}, {\em least-element list}, and {\em $s-t$ cut}  (cf. definitions in Section~\ref{sec:randomized_lb}). These bounds apply to {\em Monte Carlo randomized} algorithms as well and clearly hold also for asynchronous networks. Moreover, it is important to note that our lower bounds apply even to graphs of small diameter ($D = O(\log n)$). Furthermore we present slightly weaker lower bounds for even smaller (constant) diameters.
(Indeed, the  problems studied in this paper are ``global" problems, i.e., the network diameter of $G$ imposes an inherent lower bound on the time complexity.)

Additionally, we show that another fundamental problem, namely, the spanning tree verification problem (i.e., verifying
whether $H$ is a spanning tree) has the same lower bound of $\Omega(\sqrt{n/(B\log n)} +D)$ (cf. Section \ref{sec:deterministic_SHORT}). However,  this bound applies to only deterministic algorithms. This result strengthens the deterministic lower bound result of minimum spanning tree verification by Kor et al.~\cite{KorKP10} in that it shows that the same lower bound holds even on the simpler problem of spanning tree verification.
Moreover, we note the interesting fact that although finding a spanning tree (e.g., a breadth-first tree) can be done in $O(D)$ rounds \cite{peleg}, verifying if a given subgraph is a spanning tree requires $\tilde{\Omega}(\sqrt{n}+D)$ rounds! Thus the verification problem for spanning trees is harder than its construction in the distributed setting. This is in contrast to this well-studied problem in the centralized setting. Apart from the spanning tree verification problem, we also show deterministic lower bounds for other verification problems, including {\em Hamiltonian cycle} and {\em simple path} verification.

Our lower bounds are almost tight as we show that there exist algorithms that run in $O(\sqrt{n}\log^*n + D)$ rounds (assuming $B = O(\log n)$) for almost all the verification problems addressed here (cf.  Section~\ref{sec:tightness}).

%\medskip

\paragraph{2. Bounds on Hardness of Distributed Approximation} An important consequence of our verification lower bound is that it leads to lower bounds for exact and  approximate distributed computation. We show the unconditional time lower bound of $\Omega(\sqrt{n/(B\log n)} + D)$ for approximating many optimization problems,  including {\em MST, shortest $s-t$ path, shortest path tree}, and {\em minimum cut} (Section~\ref{sec:approxalgo}). The important point to note is that the above lower bound applies for  {\em any} approximation ratio $\alpha \geq 1$. Thus the same bound holds for exact algorithms as well (that is $\alpha = 1$). All these hardness bounds hold for randomized algorithms. (In fact, these bounds hold for {\em Monte Carlo} randomized algorithms while previous lower bounds \cite{PelegR00,Elkin06,LotkerPP06} hold only for {\em Las Vegas} randomized algorithms.)
As in our verification lower bounds, these bounds apply even to graphs of small ($O(\log n)$) diameter.
 Figure \ref{fig:lower bounds on various diameters} summarizes our lower bounds
 for various diameters.

Our results improve over previous ones (e.g., Elkin's lower bound for approximate MST and shortest path tree \cite{Elkin06}) and subsumes some well-established exact bounds (e.g., Peleg and Rubinovich  lower bound for MST
 \cite{PelegR00})  as well as show new strong bounds (both for exact and approximate computation) for many other problems (e.g., minimum cut),
  thus answering some questions that were open earlier (see the survey by Elkin \cite{Elkin-sigact04}).

The new lower bound for approximating MST simplifies and improves upon the previous $\Omega(\sqrt{n/(\alpha B \log n)} + D)$ lower bound by Elkin
%in his STOC '04 paper
\cite{Elkin06}, where $\alpha$ is the approximation factor. \cite{Elkin06} showed a {\em tradeoff} between  the running time and the approximation ratio of MST. Our result shows that approximating MST requires $\Omega(\sqrt{n/(B \log n)} + D)$ rounds, {\em regardless of $\alpha$}. Thus our result shows that there is actually
no trade-off, since there  can be no distributed approximation algorithm for MST that is significantly faster than the current exact algorithm \cite{KuttenP98, elkin-faster}, for any approximation factor $\alpha > 1$.

\begin{figure*}
  \centering
  %\scriptsize
  \begin{tabular}{c|c|c}
  \hline
  &  \footnotesize Previous lower bound for MST- & \footnotesize  New lower bound for MST-, \\
  Diameter $D$ & \footnotesize and shortest-path tree-approx. \cite{Elkin06} &  \footnotesize  shortest-path tree-approx. and \\
  & \footnotesize  (for exact algorithms, use $\alpha=1$) &  \footnotesize  all problems in Fig.~\ref{fig:all_reductions}.\\
  \hline
  $n^\delta,\ 0<\delta<1/2$  & $\Omega\left(\sqrt{\frac{n}{\alpha B}}\right)$ & $\Omega\left(\sqrt{\frac{n}{B}}\right)$  \\
  $\Theta(\log n)$ & $\Omega\left(\sqrt{\frac{n}{\alpha B\log n}}\right)$ & $\Omega\left(\sqrt{\frac{n}{B\log n}}\right)$\\
  Constant $\geq 3$  & $\Omega\left((\frac{n}{\alpha B})^{\frac{1}{2}-\frac{1}{2D-2}}\right)$ & $\Omega\left((\frac{n}{B})^{\frac{1}{2}-\frac{1}{2D-2}}\right)$\\
  4 & $\Omega\left((\frac{n}{\alpha B})^{1/3}\right)$ & $\Omega\left((\frac{n}{B})^{1/3}\right)$\\
  3 & $\Omega\left((\frac{n}{\alpha B})^{1/4}\right)$ & $\Omega\left((\frac{n}{B})^{1/4}\right)$\\
  \hline
\end{tabular}
  \caption{Lower bounds of randomized $\alpha$-approximation algorithms on graphs of various diameters. Bounds in the first column are for the MST and shortest path tree problems~\cite{Elkin06} while those in the second column are for these problems and many problems listed in Figure~\ref{fig:all_reductions}. We note that these bounds almost match the $O(\sqrt{n}\log^*{n}+D)$ upper bound for the MST problem~\cite{peleg-mst,KuttenP98} and are independent of the approximation-factor $\alpha$. Also note a simple observation that lower bounds for graphs of diameter $D$ also hold for graphs of larger diameters.}\label{fig:lower bounds on various diameters}
\end{figure*}

%\paragraph{Terminology.}
%The length of a path $p$ in $G$ is the number of edges it contains. The {\em distance} between two vertices is the length of the shortest path connecting them. The {\em diameter} of $G$, denoted by $D$, is the maximum distance between any two vertices of $G$.

\subsection{Overview of Technical Approach}

We prove our  lower bounds  by establishing an interesting connection between
communication complexity and distributed computing.
Our lower bound proofs consider the family of graphs evolved through a series of papers in the literature~\cite{Elkin06,LotkerPP06,PelegR00}.
However, while previous results \cite{PelegR00,Elkin06,LotkerPP06,KorKP10} rely on counting the number of states needed to solve the {\em mailing problem} (along with some sophisticated techniques for its variant, called {\em corrupted mailing problem}, in the case of approximation algorithm lower bounds) and use Yao's method~\cite{Yao77} (with appropriate input distributions) to get lower bounds for randomized algorithms, our results are achieved using a few steps of simple reductions, starting from problems in communication complexity, as follows (also see Figure~\ref{fig:all_reductions} for details).

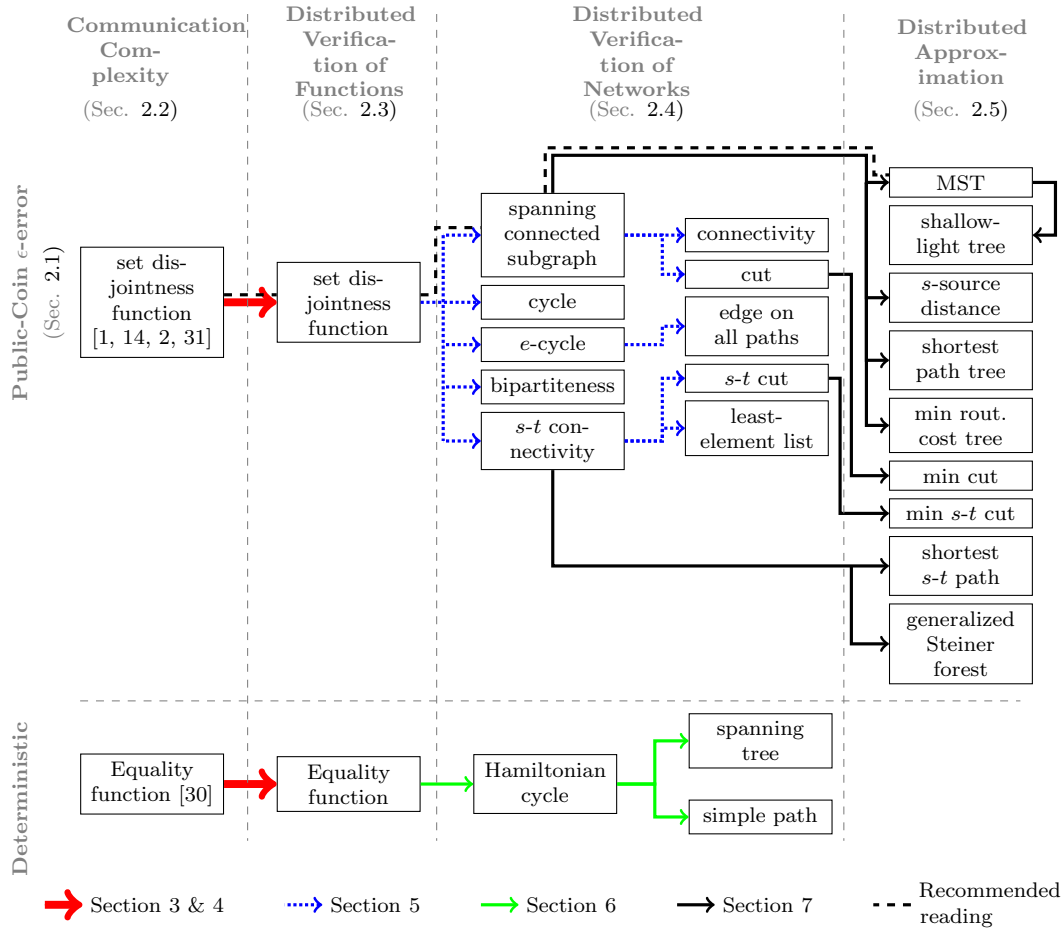
\begin{figure}
\begin{tikzpicture}[auto, node distance=0.1cm]
    \scriptsize
    % Place nodes
    \tikzstyle{mybox}=[rectangle, draw=black, text centered, text=black, text width=1.7cm]
    %\tikzstyle{myarrow}=[->, >=open triangle 90, thick]
    %\tikzstyle{line}=[-, thick]
%
    \node (span) [mybox] {spanning connected subgraph};
    \node (cycle) [mybox, below=of span] {cycle};
    \node (ecycle) [mybox, below=of cycle] {$e$-cycle};
    \node (bipartite) [mybox, below=of ecycle] {bipartiteness};
    \node  (stconn) [mybox, below=of bipartite]  {$s$-$t$ connectivity};
    \node (disj) [mybox, left=0.8 of cycle] {set disjointness function};
%
%    \draw[-, very thick, blue, densely dotted] (disj.east) -- ++(0.5,0);
    \draw[<-, very thick, blue, densely dotted] (stconn.west) -- ++(-0.5, 0);
    \draw[->, very thick, blue, densely dotted] (disj.east) -- (cycle.west);% -- ++(-0.5, 0);
    \draw[<-, very thick, blue, densely dotted] (ecycle.west) -- ++(-0.5, 0);
    \draw[<-, very thick, blue, densely dotted] (bipartite.west) -- ++(-0.5, 0);
    %\draw[<->, very thick, blue, densely dotted] (span.west) -- ++(-0.5,0) |- (stconn.west);
    \draw[<-, very thick, blue, densely dotted] (span.west) -- ++(-0.5, 0);
    \draw[-, very thick, blue, densely dotted] ([xshift=-0.5cm]span.west) -- ([xshift=-0.5cm]stconn.west);
%
%%%%%%%%%%%%%%%%%%%%%%%%%%%%%%%%%%%%%%%%%%%%%%%%%
%
%
    %\node (kcomponent) [mybox, right=0.8 of span] {$k$-component};
    \node (conn) [mybox, right=0.8 of span] {connectivity};
    \node (cut) [mybox, below= of conn] {cut};
    \node (edgeallpaths) [mybox, below= of cut] {edge on all paths};
    \node (stcut) [mybox, below= of edgeallpaths] {$s$-$t$ cut};
    \node (lelist) [mybox, below= of stcut] {least-element list};
    \node (shtree) [mybox, right=0.8 of conn] {shallow-light tree};
    \node (mst) [mybox, above= of shtree] {MST};
    \node (sdist) [mybox, below= of shtree] {$s$-source distance};
    \node (sptree) [mybox, below= of sdist] {shortest path tree};
    \node (minrtree) [mybox, below= of sptree] {min rout. cost tree};
    \node (mincut) [mybox, below= of minrtree] {min cut};
    \node (minstcut) [mybox, below= of mincut] {min $s$-$t$ cut};
    \node (stpath) [mybox, below= of minstcut] {shortest $s$-$t$ path};
    \node (steiner) [mybox, below= of stpath] {generalized Steiner forest};
    \draw[->, very thick, black!100] (span.north) -- ++(0, 0.5) -| ([xshift=-0.3cm] mst.west) -- (mst.west);
    \draw[->, very thick, black!100]  ([xshift=-0.3cm] mst.west) |- (sdist.west);
    \draw[->, very thick, black!100]  ([xshift=-0.3cm] mst.west) |- (sptree.west);
    \draw[->, very thick, black!100]  ([xshift=-0.3cm] mst.west) |- (minrtree.west);
    \draw[->, very thick, black!100] (mst.east) -- ++(0.3, 0) |- (shtree.east);
    \draw[->, very thick, blue, densely dotted] (span.east) -- ++(0.5, 0) |- (conn.west);
%    \draw[->, very thick, blue] (span.east) -- ++(0.5, 0) |- (kcomponent.west);
    \draw[->, very thick, blue, densely dotted] (span.east) -- ++(0.5, 0) |- (cut.west);
    \draw[->, very thick, blue, densely dotted] (stconn.east) -- ++(0.5, 0) |- (stcut.west);
    \draw[->, very thick, blue, densely dotted] (stconn.east) -- ++(0.5, 0) |- (lelist.west);
    \draw[->, very thick, blue, densely dotted] (ecycle.east) -- ++(0.5, 0) |- (edgeallpaths.west);
    \draw[->, very thick, black!100] (cut.east) -- ++(0.3, 0) |- (mincut.west);
    \draw[->, very thick, black!100] (stcut.east) -- ++(0.15, 0) |- (minstcut.west);
    \draw[->, very thick, black!100] (stconn.south) -- ++(0, -0.5) |- (stpath.west);
    \draw[->, very thick, black!100] ([xshift=-0.5cm] stpath.west) |- (steiner.west);
    \node (ham_func) [mybox, below=5.5 of disj] {Equality function};
    \node (ham) [mybox, right=0.7 of ham_func] {Hamiltonian cycle};
    \node (blank) [rectangle, text height=0.0cm, text width=2cm, right=0.8 of ham] {};
    \node (spt) [mybox, above= of blank] {spanning tree};
    \node (simplepath) [mybox, below= of blank] {simple path};
    \draw[->, very thick, green] (ham.east) -- ++(0.5, 0) |- (spt.west);
    \draw[->, very thick, green] (ham.east) -- ++(0.5, 0) |- (simplepath.west);
    \draw[->, very thick, green] (ham_func.east) -- (ham.west);
    \node (disj_cc) [mybox, left=0.7 of disj] {set disjointness function \cite{BabaiFS86,KalyanasundaramS92,Bar-YossefJKS04,Razborov92} };
    \node (ham_cc) [mybox, left=0.7 of ham_func] {Equality function \cite{RazS93}};
    \draw[->, very thick, red, double=red] (disj_cc.east) -- (disj.west);
    \draw[->, very thick, red, double=red] (ham_cc.east) -- (ham_func.west);
%
% Dashed lines
    \node (approximation) [above=1 of mst, text width=1.7cm, text centered, text=gray] {{\bf Distributed Approximation}};
    \node (approximation_def) [below=0 of approximation.south, text width=1.7cm, text centered, text=gray] {{\scriptsize (Sec. \ref{subsec:approx_def})}};
    \draw [-, dashed, gray] ([xshift=-0.6cm, yshift=0.5cm] approximation.west) -- ++(0, -11);
    \node (graph_verification) [left=2.4 of approximation, text width=1.7cm, text centered, text=gray] {\bf Distributed Verification of Networks};
    \node (graph_verification_def) [left=2.4 of approximation_def, text width=1.7cm, text centered, text=gray] {{\scriptsize (Sec. \ref{subsec:verification_def})}};
    \draw [-, dashed, gray] ([xshift=-1.7cm, yshift=0.5cm] graph_verification.west) -- ++(0, -11);
    \node (func_verification) [left=1.9 of graph_verification, text width=1.7cm, text centered, text=gray] {\bf Distributed Verification of Functions};
    \node (func_verification_def) [left=1.9 of graph_verification_def, text width=1.7cm, text centered, text=gray] {{\scriptsize (Sec. \ref{subsec:func_verification_def})}};
    \draw [-, dashed, gray] ([xshift=-0.4cm, yshift=0.5cm] func_verification.west) -- ++(0, -11);
    \node (communication) [left=1 of func_verification, text width=1.7cm, text centered, text=gray] {\bf Communication Complexity};
    \node (communication_def)[left=1 of func_verification_def, text width=1.7cm, text centered, text=gray] {{\scriptsize (Sec. \ref{subsec:compcomp})}};
    \draw [-, dashed, gray] ([yshift=0.7cm] ham_cc.north west) -- ++(12.5, 0);
    \node (error2) [rotate=90, left=0.8 of disj_cc.north west, xshift=1cm, text=gray, text centered, text width=3cm]   {\bf Public-Coin $\epsilon$-error};
    \node (error2_def) [rotate=90, below=0 of error2, text=gray, text centered]   {(\scriptsize Sec. \ref{subsec:error_def})};
    \node (error1) [rotate=90, left=0.8 of ham_cc.north west, xshift=1cm, text=gray, text centered, text width=3cm]   {\bf  Deterministic};
    %\node (error1_def) [rotate=90, below=0 of error1, text=gray, text centered]   {(\scriptsize Sec. \ref{subsec:error_def})};
%
    \node (blank1) [below=1.5 of ham_cc.east, text width=0cm] {};
    \node (cc_veri)[left=0 of blank1.east, text centered] {Section~\ref{sec:communication_complexity} \& \ref{sec:function_verification}};
    \draw [<-, very thick, red, double=red] (cc_veri.west) -- ++(-0.5, 0);
    \node (rand2)[right=1.2 of cc_veri.east, text centered] {Section~\ref{sec:randomized_lb}};
    \draw [<-, very thick, blue, densely dotted] (rand2.west) -- ++(-0.5, 0);
    \node (rand1)[right=1.2 of rand2.east, text centered] {Section~\ref{sec:deterministic_SHORT}};
    \draw [<-, very thick, green] (rand1.west) -- ++(-0.5, 0);
    \node (approx)[right=1.2 of rand1.east, text centered] {Section~\ref{sec:approxalgo}};
    \draw [<-, very thick, black] (approx.west) -- ++(-0.5, 0);
    \node (recommend)[right=1.2 of approx.east, text centered, text width = 1cm] {Recommended reading};
    \draw [-, very thick, dashed, black] (recommend.west) -- ++(-0.5, 0);
%
% Recommended reading
%
    \draw[-, very thick, dashed, black] ([yshift=0.1cm] disj_cc.east) -- ([yshift=0.1cm]disj.west);
    \draw[-, very thick, dashed, black] ([yshift=0.1cm] disj.east) -- ++(0.2,0) |-([yshift=0.1cm] span.west);
    \draw[-, very thick, dashed, black] ([xshift=-0.1cm] span.north) -- ++(0, 0.6) -| ([xshift=-0.2cm, yshift=0.1cm] mst.west) -- ([yshift=0.1cm] mst.west);

\end{tikzpicture}
\caption{Problems and reductions between them to obtain randomized and deterministic lower bounds. For all problems, we obtain lower bounds as in Figure~\ref{fig:lower bounds on various diameters}. In order to get the whole picture of the paper, we recommend reading along the black dashed line.
Definitions of (Monte Carlo) randomized algorithms can be found in Section~\ref{subsec:error_def}. Definitions of problems in communication complexity, distributed verification of functions, distributed verification of networks and distributed approximation, can be found in Section \ref{subsec:compcomp}, \ref{subsec:func_verification_def}, \ref{subsec:verification_def} and \ref{subsec:approx_def}, respectively. }\label{fig:all_reductions}
\end{figure}

(Section~\ref{sec:communication_complexity}) First, we reduce the lower bounds of problems in the standard communication complexity model~\cite{KNbook} to the lower bounds of the equivalent problems in the ``distributed version'' of communication complexity.
Specifically, we prove the {\em Simulation Theorem} (cf. Section~\ref{sec:communication_complexity}) which relates the {\em communication} lower bound from the standard communication complexity model \cite{KNbook} to compute  some appropriately chosen function $f$, to the distributed  {\em time} complexity lower bound for computing the same function
in a specially chosen graph $G$. In the standard model, Alice and Bob can communicate directly (via a bidirectional edge of bandwidth one).
In the distributed model, we assume that Alice and Bob are some vertices of  $G$ and they together wish to compute the function $f$ using the communication graph $G$. The choice of graph $G$ is critical.
We use a graph called $\graph$ (parameterized by $\Gamma$, $d$ and $p$) that was first used in \cite{Elkin06}.  We show a reduction from the standard model to the distributed model, the proof of which relies on some observations used in previous results (e.g., \cite{PelegR00}).

(Section~\ref{sec:function_verification}) The connection established in the first step allows us to bypass the state counting argument and Yao's method, and reduces our task in proving lower bounds of verification problems to merely picking the right function $f$ to reduce from.
The function $f$ that is useful in showing our randomized lower bounds is the {\em set disjointness function} \cite{BabaiFS86,KalyanasundaramS92,Bar-YossefJKS04,Razborov92}, which is the quintessential problem in the world of communication complexity with applications to diverse areas and has been studied for decades (see a recent survey in \cite{setdisj-survey}).
Following a result well known in communication complexity \cite{KNbook}, we show that the distributed version of this problem has an $\Omega(\sqrt{n/(B\log n)})$ lower bound on graphs of small diameter.

(Section~\ref{sec:randomized_lb} \& \ref{sec:deterministic_SHORT}) We then reduce this problem to the verification problems using simple reductions similar to those used in data streams~\cite{HenzingerRR98}. The set disjointness function yields randomized lower bounds and works for many problems (see Figure~\ref{fig:all_reductions}), but it does not reduce to certain other problems such as spanning tree. To show lower bounds for these other problems, we use a different function $f$ called {\em equality function}. However, this reduction yields only deterministic lower bounds for the corresponding verification problems.
%The spanning tree result improves and simplifies a result in \cite{KorKP10}.

(Section~\ref{sec:approxalgo}) Finally, we reduce the verification problem to hardness of distributed approximation for a variety of problems to show that the same lower bounds hold for approximation algorithms as well. For this,  we use a reduction whose idea is similar to one  used to prove hardness of approximating TSP (Traveling Salesman Problem) on general graphs (see, e.g., \cite{Vazirani-book}): We convert a verification problem to an optimization problem by introducing edge weights in such a way that there is a large gap between the optimal values for the cases where $H$ satisfies, or does not satisfy a certain property. This technique
is surprisingly simple, yet yields strong unconditional hardness bounds ---
many hitherto unknown, left open (e.g., minimum cut) \cite{Elkin-sigact04} and some that improve
over known ones (e.g., MST and shortest path tree)~\cite{Elkin06}. As mentioned earlier, our approach shows that
approximating MST by {\em any} factor needs $\tilde{\Omega}(\sqrt{n})$
time, while the previous result due to Elkin gave a bound that depends on  $\alpha$ (the approximation factor), i.e.  $\tilde{\Omega}(\sqrt{n/\alpha})$, using more sophisticated techniques.

%Elkin obtained his bound using more sophisticated techniques, in particular, he used a problem called the {\em corrupted mailing problem} and related the error rate in this problem to the approximation factor of MST.

%
Figure~\ref{fig:all_reductions} summarizes these reductions that will be proved in this paper.
Our proof technique via this approach is quite general and conceptually straightforward to apply as it hides all complexities in the well studied communication complexity. Yet, it yields tight lower bounds for many problems (we show almost matching upper bounds for many problems in Section~\ref{sec:tightness}).
It also has some advantages over the previous approaches.
First, in the previous approach, we have to start from scratch every time we want to prove a lower bound for a new problem. For example, extending from the mailing problem in \cite{PelegR00} to the corrupted mailing problem in \cite{Elkin06} requires some sophisticated techniques. Our new technique allows us to use known lower bounds in communication complexity to do such a task.
Secondly, extending a deterministic lower bound to a randomized one is sometimes difficult. As in our case, our randomized lower bound of the spanning connected subgraph problem would be almost impossible without connecting it to the communication complexity lower bound of the set disjointness problem (whose strong randomized lower bound is a result of years of studies \cite{BabaiFS86,KalyanasundaramS92,Bar-YossefJKS04,Razborov92}). One important consequence is that this technique allows us to obtain lower bounds for {\em Monte Carlo} randomized algorithms while previous lower bounds hold only for Las Vegas randomized algorithms.
We believe that this technique could lead to many new lower bounds of distributed algorithms.

\paragraph{Recent results} After the preliminary version of this paper  (\cite{DasSarmaHKKNPPW11}) appeared, the connection between communication complexity and distributed algorithm lower bounds has been further used to develop some new lower bounds. In \cite{NanongkaiDP11}, the Simulation Theorem (cf. Theorem~\ref{thm:cc_to_distributed}) is extended to show a connection between bounded-round communication complexity and  distributed algorithm lower bounds. It is then used to show a tight lower bound for distributed random walk algorithms. In \cite{FrischknechtHW11}, lower bounds of computing diameter of a network and related problems are shown by reduction from the communication complexity of set disjointness. This is done by considering the communication at the bottleneck of the network (sometimes called {\em bisection width} \cite{KNbook,Leighton91}).
A similar argument is also used in \cite{KuhnO11} to show lower bounds on directed networks.
%
%can be traced back to the original paper of Yao~\cite{Yao79-cc}\danupon{Check this} on communication complexity.

\section{Preliminaries}\label{sec:prelim}
To make it easy to look up for definitions, we collect all necessary definitions in this section. We recommend the readers to skip this section in the first read and come back when necessary.

This section is organized as follows (also see Figure~\ref{fig:all_reductions} for a pointer to a subsection for each definition). In Subsection~\ref{subsec:error_def}, we define the notion of $\epsilon$-error randomized public-coin algorithms and the worse-case running time of these algorithms. In Subsection~\ref{subsec:compcomp}, we define the communication complexity model and the set disjointness and equality problems. We then extend this model to the model of distributed verification of functions in Subsection~\ref{subsec:func_verification_def}. In Subsection~\ref{subsec:verification_def}, we give a formal definition of the distributed verification problem which we explained informally in Section~\ref{sec:intro}. We also define the specific distributed verification problems considered in this paper. Finally, in Subsection~\ref{subsec:approx_def}, we define the notion of approximation algorithms.

\subsection{Randomized (Monte Carlo) Public-Coin Algorithms and the Worst-Case Running Time}\label{subsec:error_def}
In this paper, we show lower bounds of distributed algorithms that are {\em Monte Carlo}. Recall that a Monte Carlo algorithm is a randomized algorithm whose output may be incorrect with some probability.
Formally, let $\cA$ be any algorithm for computing a function $f$. We say that $\cA$ computes $f$ with {\em $\epsilon$-error} if for every input $x$, $\cA$ outputs $f(x)$ with probability at least $1-\epsilon$. Note that a $0$-error algorithm is deterministic.

We note the fact that lower bounds of Monte Carlo algorithms also imply lower bounds of {\em Las Vegas} algorithms (whose output is always correct but the running time is only in expectation). Thus, lower bounds in this paper hold for both types of algorithms. %Throughout this paper, when we say randomized algorithms, we always refer to Monte Carlo algorithms.

\paragraph{Public coin} We say that a randomized distributed algorithm uses a {\em public coin} if all nodes have an access to a common random string (chosen according to some probability distribution).
In this paper, we are interested in the lower bounds of public-coin randomized distributed algorithms. We note that these lower bounds also imply time lower bounds of {\em private-coin} randomized distributed algorithms, where nodes do not share a random string, since allowing a public coin only gives more power to the algorithms.

\paragraph{Worst-case running time} For any public-coin randomized distributed algorithm $\cA$ on a network $G$ and input $\cal I$ (given to nodes in $G$), we define the {\em worst-case running time} of $\cA$ on input $\cal I$ to be the maximum number of rounds needed to run $\cA$ among all possible (shared) random strings. The worst-case running time of $\cA$ is the maximum, over all inputs $\cal I$, of the worst-case running time of $\cA$ on $\cal I$.

%As in the case of communication complexity, we will only consider the lower bounds of {\em public-coin randomized distributed algorithms} where all nodes in the network share a random string. Note again that, in reality, nodes might not share a random string. However, our lower bounds remain to hold since they hold even in the stronger model of public coin.

\subsection{Communication Complexity}\label{subsec:compcomp}
In this paper, we consider the standard model of communication complexity. To avoid confusion, we define the model as a special case of the distributed algorithm model. We refer to \cite{KNbook} for the conventional definition, further details and discussions.

In this model, there are two nodes in the network connected by an edge. We call one node {\em Alice} and the other node {\em Bob}. Alice and Bob each receive a $b$-bit binary string, for some integer $b\geq 1$, denoted by $x$ and $y$ respectively. Together, they both want to compute $f(x, y)$ for a Boolean function $f:\{0, 1\}^b\times \{0, 1\}^b\rightarrow \{0, 1\}$.
%
%We consider two models of communication.
%This is the standard model in communication complexity.
%
%The two parties can communicate via a bidirectional edge. We call the party receiving $x$ {\em Alice}, and the other party {\em Bob}. At the end of the process, both Alice and Bob will output $f(x, y)$. (See Fig.~\ref{fig:communication complexity}.)
%
In the end of the process, we want both Alice and Bob to know the value of $f(x, y)$. %(See Fig.~\ref{fig:communication complexity}.)
We are interested in the worst-case running time of distributed algorithms on this network when one bit can be sent on the edge in each round (thus the running time is equal to the number of bits Alice and Bob exchange).

%\begin{figure}
%\centering
%\subfigure[Communication complexity]{\includegraphics[height=3cm]{communication_complexity.eps}\label{fig:communication complexity}}
%\vspace{5pt}
%\subfigure[Verification of functions]{\includegraphics[height=3cm]{function_verification.eps}\label{fig:function verification}}
%\caption{The models of communication complexity and verification of functions. In the first model, Alice and Bob can communicate to each other directly and we are interested in the number of bits they need to exchange in order to compute $f(x, y)$. In the second model, Alice and Bob are part of the network and we are interested in the time (number of rounds) needed to compute $f(x, y)$.}
%\end{figure}

For any Boolean function $f$ and $\epsilon>0$, we let $R_{\epsilon}^{cc-pub}(f)$ denote the minimum worst-case running time of the best $\epsilon$-error randomized algorithm for computing $f$ in the communication complexity model.

In this paper, we are interested in two Boolean functions, {\em set disjointness ($\DISJfunc$)} and {\em equality ($\EQfunc$)} functions, defined as follows.

\begin{itemize}
\item {\bf Set Disjointness function ($\DISJfunc$).} Given two $b$-bit strings $x$ and $y$, the {\em set disjointness function}, denoted by $\DISJfunc(x, y)$, is defined to be one if the inner product $\langle x, y\rangle$ is $0$ (i.e., there is no $i$ such that $x_i=y_i=1$) and zero otherwise.
%We refer to the problem of computing $\DISJfunc$ function on $\ElkinGraph$ on $\Gamma$-bit input strings given to $s$ and $r$ by $\DISJ(\ElkinGraph, s, r, \Gamma)$.

\item {\bf Equality function ($\EQfunc$).} Given two $b$-bit strings $x$ and $y$, the {\em equality function}, denoted by $\EQfunc(x, y)$, is defined to be one if $x=y$ and zero otherwise.
\end{itemize}

\subsection{Distributed Verification of Functions}\label{subsec:func_verification_def}
%This model is similar to the communication complexity model except that Alice and Bob are not connected by an edge but they are instead some nodes in a large network (see Fig.~\ref{fig:function verification} as an example). The detail is as follows.
%
%
We consider the same problem as in the case of communication complexity. That is, Alice and Bob receive $b$-bit binary strings $x$ and $y$ respectively and they want to compute $f(x, y)$ for some Boolean function $f$. However, Alice and Bob are now distinct vertices in a $B$-model distributed network $G$ (cf. Section~\ref{subsec:intro:background}). We denote Alice's node (which receives $x$) by $s$ and Bob's node (which receives $y$) by $r$. At the end of the process, both $s$ and $r$ will output $f(x, y)$. We are interested in the worst-case running time a distributed algorithm needs in order to compute function $f$.

For any network $G$ (with two nodes marked as $s$ and $r$), Boolean function $f$ and $\epsilon>0$, we let $R_{\epsilon}^{G}(f)$ denote the worst-case running time of the best $\epsilon$-error randomized distributed algorithm for computing $f$ on $G$.

%For any class of networks $\cal G$, we let $R_\epsilon^{\mathcal G}(f)$ be the worst $R_\epsilon^{G}(f)$ among all $G\in\mathcal{G}$.

In this model, we consider the set disjointness and equality functions as in the communication complexity model (cf. Subsection~\ref{subsec:compcomp}).

\subsection{Distributed Verification of Networks}\label{subsec:verification_def} We already gave an informal definition of this problem in Section~\ref{sec:intro}. We now define the problem formally. In the distributed network $G$, we describe its subgraph $H$ as an input as follows. Each node $v$ in $G$ with neighbors $u_1,\dots, u_{d(v)}$, where $d(v)$ is the degree of $v$, has $d(v)$ {\em Boolean indicator variables}  $Y_v(u_1), \dots, Y_v(u_{d(v)})$ indicating which of the edges incident to $v$  participate in the subgraph $H$.
The indicator variables must be consistent, i.e., for every edge $(u,v)$, $Y_v(u)=Y_u(v)$ (this is easy to verify locally with a single round of communication).

Let $H_Y$ be the set of edges whose indicator variables are $1$; that is,
$$H_Y=\{(u, v)\in E \mid Y_u(v)=1\}.$$
Given a predicate $\Pi$ (which may specify statements such as ``$H_Y$ is connected'' or ``$H_Y$ is a spanning tree" or ``$H_Y$ contains a cycle''), the output for a verification problem at each vertex $v$ is an assignment to a (Boolean) output variable $A^v$,
where $A^v=1$ if $H_Y$  satisfies the predicate $\Pi$, and $A^v=0$ otherwise.

We say that a distributed algorithm $\cA_{\Pi}$ verifies predicate $\Pi$ if, for every graph $G$ and subgraph $H_Y$ of $G$, all nodes in $G$ knows whether $H_Y$ satisfies $\Pi$ after we run $\cA_{\Pi}$; that is,
after the execution of $\cA_{\Pi}$ on graph $G$, at each vertex $v$ the output variable $A^v$ is one if $H_Y$ satisfies predicate $\Pi$, and zero otherwise. Note again that the time complexity
of the verification algorithm is measured with respect to the size and diameter of $G$ (independently from $H_Y$). When $Y$ is clear from the context, we use $H$ to denote $H_Y$.

%\subsection{Graph verification problems}

We now define problems considered in this paper.
%We first define problems that will be shown in Section~\ref{sec:randomized_lb} to have randomized lower bounds (see Fig.~\ref{fig:all_reductions}).

\begin{itemize}
\item
{\bf connected spanning subgraph verification:} We want to verify whether $H$ is connected and spans all nodes of $G$, i.e., every node in $G$ is incident to some edge in $H$.

\item
{\bf cycle containment verification:} We want to verify if $H$ contains a cycle. %(Section~\ref{sec:reduction from disj}.)

\item
{\bf $e$-cycle containment verification:} Given an edge $e$ in $H$ (known to vertices adjacent to it), we want to verify if $H$ contains a cycle containing $e$. %(Section~\ref{sec:reduction from disj}.)

\item
{\bf bipartiteness verification}: We want to verify whether $H$ is bipartite. %(Section~\ref{sec:reduction from disj}.)

\item
{\bf $s$-$t$ connectivity verification}: In addition to $G$ and $H$, we are given two vertices $s$ and $t$ ($s$ and $t$ are known by every vertex). We would like to verify whether $s$ and $t$ are in the same connected component of $H$. %(Section~\ref{sec:st_lowerbound}.)

\item
{\bf connectivity verification}: We want to verify whether $H$ is connected.
%We also consider the {\bf $k$-component verification problem} where we want to verify whether $H$ has at most $k$ connected components. (Note that $k$ is not part of the input so $2$-component and $3$-component problems are different problems.) The connectivity verification problem is the special case where $k=1$. %(Section \ref{sec:rand reduction from other problems})

%%\item{\bf $k$-component verification problem}: We want to verify whether $H$ has at most $k$ connected component. (Note that $k$ is not part of the input so $2$-component and $3$-component problems are different problems.) \danupon{Connectivity is the special case where $k=1$.}

\item
{\bf cut verification:} We want to verify whether $H$ is a cut of $G$, i.e., $G$ is not connected when we remove edges in $H$. %(Section \ref{sec:rand reduction from other problems})

\item
{\bf edge on all paths verification:} Given two nodes $u$, $v$ and an edge $e$. We want to verify whether $e$ lies on all paths between $u$ and $v$ in $H$. In other words, $e$ is a $u$-$v$ cut in $H$. %(Section \ref{sec:rand reduction from other problems})

\item
{\bf $s$-$t$ cut verification}: We want to verify whether $H$ is an $s$-$t$ cut, i.e., when we remove all edges $E_H$ of $H$ from $G$, we want to know whether $s$ and $t$ are in the same connected component or not. %(Section \ref{sec:rand reduction from other problems})

\item
{\bf least-element list verification~\cite{cohen,KhanKMPT08}:} The input of this problem is different from other problems and is as follows. Given a distinct rank (integer) $r(v)$ to each node $v$ in the weighted graph $G$, for any nodes $u$ and $v$, we say that $v$ is the {\em least element} of $u$ if $v$ has the lowest rank among vertices of distance at most $d(u, v)$ from $u$. Here, $d(u, v)$ denotes the weighted distance between $u$ and $v$. The {\em Least-Element List} (LE-list) of a node $u$ is the set $\{\langle v, d(u, v)\rangle \mid \mbox{$v$ is the least element of u}\}$. %(Section \ref{sec:rand reduction from other problems})

In the least-element list verification problem, each vertex knows its rank as an input, and some vertex $u$ is given a set $S=\{\langle v_1, d(u, v_1)\rangle, \langle v_2, d(u, v_2)\rangle, \ldots \}$ as an input. We want to verify whether $S$ is the least-element list of $u$. %(Section \ref{sec:rand reduction from other problems})

%\end{itemize}
%
%We now define problems that will be shown in Section~\ref{sec:dete} to have deterministic randomized lower bounds (see Figure~\ref{fig:all_reductions}).
%
%
%\begin{itemize}

\item
{\bf Hamiltonian cycle verification:} We would like to verify whether $H$ is a Hamiltonian cycle of $G$, i.e., $H$ is a simple cycle of length $n$.

\item
{\bf spanning tree verification:} We would like to verify whether $H$ is a tree spanning $G$.
%We note that this problem is a special case of the minimum spanning tree verification problem~\cite{KorKP10} and thus our lower bound for the spanning tree verification problem implies a stronger lower bound for the minimum spanning tree verification problem.

\item
{\bf simple path verification:} We would like to verify that $H$ is a simple path, i.e., all nodes have degree either zero or two in $H$ except two nodes that have degree one and there is no cycle in $H$.
\end{itemize}

\subsection{Approximation Algorithms}\label{subsec:approx_def}
In a graph optimization problem $\cP$ in a distributed network, such as finding a MST, we are given a non-negative weight $\omega(e)$ on each edge $e$ of the network (each node knows the weights of all edges incident to it). Each pair of network and weight function $(G, \omega)$ comes with a nonempty set of {\em feasible solutions} for a problem $\cP$; e.g., for the case of finding a MST, all spanning trees of $G$ are feasible solutions. The goal of $\cP$ is to find a feasible solution that minimizes or maximizes the total weight. We call such a solution an {\em optimal solution}. For example, a spanning tree of minimum weight is an optimal solution for the MST problem.

For any $\alpha\geq 1$, an {\em $\alpha$-approximate solution} of $\cP$ on weighted network $(G, \omega)$ is a feasible solution whose weight is not more than $\alpha$ (respectively,  $1/\alpha$) times of the weight of the optimal solution of $\cP$ if $\cP$ is a minimization (respectively, maximization) problem. We say that an algorithm $\cA$ is an $\alpha$-approximation algorithm for problem $\cP$ if it outputs an $\alpha$-approximate solution for any weighted network $(G, \omega)$.
In case of randomized algorithms (cf. Subsection~\ref{subsec:error_def}), we say that an $\alpha$-approximation $T$-time algorithm is $\epsilon$-error if it outputs an answer that is not $\alpha$-approximate with probability at most $\epsilon$ and always finishes in time $T$, regardless of the input and the choice of random string.
%
%Thus, when we say that a randomized $\alpha$-approximation algorithm finishes in $O(\sqrt{n})$ time then we mean that it always finishes in $O(\sqrt{n})$ regardless of the input and the choice of random string; however, it may output an answer that is not $\alpha$-approximate with some small probability.

In this paper, we consider the following problems.

\begin{itemize}
\item In the {\bf minimum spanning tree} problem~\cite{Elkin06,PelegR00}, we want to compute the weight of the minimum spanning tree (i.e., the spanning tree of minimum weight). In the end of the process all nodes should know this weight.

\item Consider a network with two cost functions associated to edges, weight and length, and a root node $r$. For any spanning tree $T$, the radius of $T$ is the maximum length (defined by the length function) between $r$ and any leaf node of $T$. Given a root node $r$ and the desired radius $\ell$, a {\bf shallow-light tree}~\cite{peleg} is the spanning tree whose radius is at most $\ell$ and the total weight is minimized (among trees of the desired radius).

\item Given a node $s$, the {\bf $s$-source distance} problem~\cite{Elkin05} is to find the distance from $s$ to every node. In the end of the process, every node knows its distance from $s$.

\item In the {\bf shortest path tree} problem~\cite{Elkin06}, we want to find the shortest path spanning tree rooted at some input node $s$, i.e., the shortest path from $s$ to any node $t$ must have the same weight as the unique path from $s$ to $t$ in the solution tree. In the end of the process, each node should know which edges incident to it are in the shortest path tree.

\item The {\bf minimum routing cost spanning tree} problem~\cite{WuLBCRT99} is defined as follows. We think of the weight of an edge as the cost of routing messages through this edge. The routing cost between any node $u$ and $v$ in a given spanning tree $T$, denoted by $c_T(u, v)$, is the distance between them in $T$. The routing cost of the tree $T$ itself is the sum over all pairs of vertices of the routing cost for the pair in the tree, i.e., $\sum_{u, v\in V} c_T(u, v)$. Our goal is to find a spanning tree with minimum routing cost.

\item A set of edges $E'$ is a {\bf cut} of $G$ if $G$ is not connected when we delete $E'$. The {\bf minimum cut} problem~\cite{Elkin-sigact04} is to find a cut of minimum weight. A set of edges $E'$ is an {\em $s$-$t$ cut} if there is no path between $s$ and $t$ when we delete $E'$ from $G$.
    The {\bf minimum $s$-$t$ cut} problem is to find an $s$-$t$ cut of minimum weight.
    %In the end of the process, nodes must know the weight of the minimum cut and minimum $s$-$t$ cut.
    %

\item Given two nodes $s$ and $t$, the {\bf shortest $s$-$t$ path} problem is to find the length of the shortest path between $s$ and $t$.

\item The {\bf generalized Steiner forest} problem~\cite{KhanKMPT08} is defined as follows. We are given $k$ disjoint subsets of vertices $V_1, ..., V_k$ (each node knows which subset it is in). The goal is to find a minimum weight subgraph in which each pair of vertices belonging to the same subsets is connected. In the end of the process, each node knows which edges incident to it are in the solution.

%\item Given a distinct rank (integer) $r(v)$ to each node $v$ in the weighted graph $G$, for any nodes $u$ and $v$, we say that $v$ is the {\em least element} of $u$ if $v$ has the lowest rank among vertices of distance at most $d(u, v)$ from $u$. Here, $d(u, v)$ denotes the weighted distance between $u$ and $v$. The {\em Least-Element List} (LE-list) of a node $u$ is the set $\{<v, d(u, v)>\ |\ $ $v$ is the least element of $u$ $\}$.\danupon{We have to think where to fit this problem in!!!}
\end{itemize}

Note that in the minimum spanning tree, minimum cut, minimum $s$-$t$ cut, and shortest $s$-$t$ path problems, an $\alpha$-approximation algorithm should find a solution that has total weight at most $\alpha$ times the weight of the optimal solution. For the $s$-source distance problem, an $\alpha$-approximation algorithm should find an approximate distance $d(v)$ of every vertex $v$ such that $distance(s, v) \leq d(v)\leq \alpha\cdot distance(s, v)$ where $distance(s, v)$ is the distance of $s$ from $v$. Similarly, an $\alpha$-approximation algorithm for the shortest path tree problem should find a spanning tree $T$ such that, for any node $v$, the length $\ell$ of a unique path from $s$ to $v$ in $T$ satisfies $\ell\leq \alpha \cdot distance(s, v)$.

%\newpage
\section{From Communication Complexity to Distributed Computing}\label{sec:communication_complexity}

%%%%%%%%%%%%%%%%%%%%%%%%%%%%%%%%%%%%%%%%%%%%%%%%%%%%%%%%%%%%%%%%
%%%%%%%%%%%%%%%%%%%%%%%%%%%%%%%%%%%%%%%%%%%%%%%%%%%%%%%%%%%%%%%%
%%%%%%%%%%%%%%%%%%%%%%%%%%%%%%%%%%%%%%%%%%%%%%%%%%%%%%%%%%%%%%%%
%
%
In this section, we show a connection between the communication complexity model (cf. Section~\ref{subsec:compcomp}) and the model of distributed verification of functions (cf. Section~\ref{subsec:func_verification_def}) on a family of graphs called $\graph$. This family of graphs was first defined in \cite{Elkin06} (which was extended from \cite{PelegR00}). We will define this graph in Subsection~\ref{sec:ElkinGraph} for completeness.

The main result of this section shows that if there is a fast $\epsilon$-error algorithm for computing $f$  on $\graph$, then there is a fast $\epsilon$-error algorithm for Alice and Bob to compute $f$ in the communication complexity model. We call this the {\em Simulation Theorem}.
We state the theorem below. The rest of this section is devoted to define the graph $\graph$ and to prove the theorem.

\begin{theorem}[Simulation Theorem]\label{thm:cc_to_distributed}
For any $\Gamma$, $d$, $p$, $B$, $\epsilon\geq 0$, and function $f:\{0, 1\}^{b}\times \{0, 1\}^{b} \rightarrow \{0, 1\}$, if there is an $\epsilon$-error distributed algorithm on $\graph$ that computes $f$ faster than $\frac{d^p-1}{2}$ time, i.e.,
\[R_\epsilon^{\ElkinGraph}(f) < \frac{d^p-1}{2}\]
then there is an $\epsilon$-error algorithm in the communication complexity model that computes $f$ in at most $2dpBR_\epsilon^{\ElkinGraph}(f)$ time.
In other words,
\[R_\epsilon^{cc-pub}(f)\leq 2dpBR_\epsilon^{\ElkinGraph}(f)\,.\]
\end{theorem}

We first describe the graph $\graph$ with parameters $\Gamma$, $d$ and $p$ and distinct vertices $s$ and $r$.

\subsection{Description of $\graph$~\cite{Elkin06}} \label{sec:ElkinGraph}

%\begin{wrapfigure}{r}{0.5\linewidth}
\begin{figure}%[tpb]
  \centering
 {\scriptsize
    \psfrag{A}[c]{$\cP^1$}
    \psfrag{B}[c]{$\cP^\ell$}
    \psfrag{C}[c]{$\cP^{\Gamma}$}
    \psfrag{a}[c]{$v_0^1$}
    \psfrag{b}[c]{$v_1^1$}
    \psfrag{c}[c]{$v_2^1$}
    \psfrag{d}{$v_{d^p-1}^1$}
    \psfrag{e}[c]{$v_0^\ell$}
    \psfrag{f}[c]{$v_1^\ell$}
    \psfrag{g}[c]{$v_2^\ell$}
    \psfrag{h}{$v_{d^p-1}^\ell$}
    \psfrag{i}[c]{$v_0^{\Gamma}$}
    \psfrag{j}[c]{$v_1^{\Gamma}$}
    \psfrag{k}[c]{$v_2^{\Gamma}$}
    \psfrag{l}[l]{$v_{d^p-1}^{\Gamma}$}
    \psfrag{s}[r]{$s=u^p_0$}
    \psfrag{r}{$u^p_{d^p-1}=r$}
    \psfrag{m}[c]{$u^p_{1}$}
    \psfrag{n}[c]{$u^p_{2}$}
    \psfrag{o}{$u^0_0$}
    \psfrag{p}{$u^{p-1}_0$}
    \psfrag{q}{$u^{p-1}_1$}
    \psfrag{D}{$R_2$}
    \psfrag{E}{$R_1$}
    \psfrag{F}{$R_0$}
    %\subfigure[Example of $\ElkinGraph$ (here $d=2$)]{
    \begin{center}
%    \vspace{4mm}
    \includegraphics[width=0.4\linewidth]{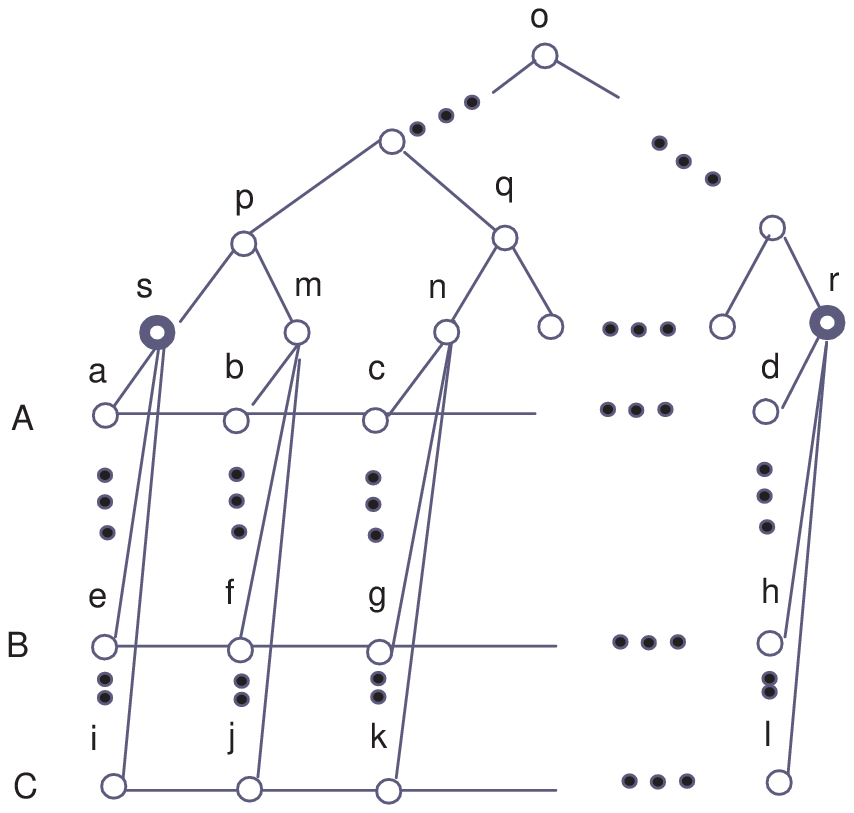}
    \end{center}
    %}
    }
    \caption{An example of $\ElkinGraph$ (here $d=2$).}\label{fig:GraphExample}\label{fig:ElkinGraph}
%    \vspace{4mm}
\end{figure}
%\end{wrapfigure}

We now describe the network $\graph$ in detail. The two basic units in the construction are {\em paths} and a {\em tree}. There are $\Gamma$ paths, denoted by $\cP^1, \cP^2, \ldots, \cP^{\Gamma}$, each having $d^p$ nodes, i.e., for $\ell=1, 2, \ldots \Gamma$,
%$V(\cP^\ell)= \{v_0^\ell,\dots,v_{d^p-1}^\ell\}$ and $E(\cP^\ell) = \{(v_i^\ell, v_{i+1}^\ell) \mid 0 \le i < d^p-1\}.$
%
% REMOVE TO SAVE SPACE
$$V(\cP^\ell)= \{v_0^\ell,\dots,v_{d^p-1}^\ell\}
~~~~~\mbox{and}~~~~~
E(\cP^\ell) = \{(v_i^\ell, v_{i+1}^\ell) \mid 0 \le i < d^p-1\}\,.$$
There is a tree, denoted by $\cT$ having depth $p$ where each non-leaf node has $d$ children (thus, there are $d^p$ leaf nodes). We denote the nodes of $\cT$ at level $\ell$ from left to right by $u^\ell_0, \ldots, u^\ell_{d^\ell-1}$ (so, $u^0_0$ is the root of $\cT$ and $u^p_0, \ldots, u^p_{d^p-1}$ are the leaves of $\cT$). %
For any $\ell$ and $j$, the leaf node $u^p_j$ is connected to the corresponding path node $v_{j}^{\ell}$ by a {\em spoke} edge $(u^p_j, v_j^\ell)$. Finally, we set the two special nodes (which will receive input strings $x$ and $y$) as $s=u_0^p$ and $r=u_{d^p-1}^p$.
%
%Finally, we create infinitely many copies of all edges except those in $\cT$.
%
Figure~\ref{fig:ElkinGraph} depicts this network. We note the following lemma proved in \cite{Elkin06}.

%%%%%%%%%%%%%%%%%%%%%%%%%%%%%%%

\begin{lemma}\label{lem:size}\cite{Elkin06}
The number of vertices in $\graph$ is $n = \Theta(\Gamma d^p)$ and its diameter is $2p+2$.
\end{lemma}

\subsection{Terminologies}
For any $1\leq i\leq \lfloor (d^p-1)/2\rfloor$, define the {\em $i$-left} and the {\em $i$-right} of path $\mathcal{P}^\ell$ as
%
%$L_i(\mathcal{P}^\ell)=\{v^\ell_j\ |\ j \leq d^p-1-i\}$ and $R_i(\mathcal{P}^\ell)=\{v^\ell_j\ |\ j \geq i\}$,
%
%%%%%%%%Use above to save space %%%%%%%%
\[L_i(\mathcal{P}^\ell)=\{v^\ell_j\ |\ j \leq d^p-1-i\}%, ~~\mbox{and}\]
~~~~~\mbox{and}~~~~~
R_i(\mathcal{P}^\ell)=\{v^\ell_j\ |\ j \geq i\}\,, \]
respectively. Thus, $L_0(\cP^\ell)=R_0(\cP^\ell)=V(\cP^\ell)$. Define the $i$-left of the tree $\cT$, denoted by $L_i(\cT)$, as the union of the set $S=\{u^p_j\ |\ j \leq d^p-1-i\}$ and all ancestors of all vertices in $S$. Similarly, the $i$-right $R_i(\cT)$ of the tree $\cT$ is the union of set $S=\{u^p_j\ |\ j \geq i\}$ and all ancestors of all vertices in $S$.
Now, the {\em $i$-left} and {\em $i$-right} sets of $\graph$ are the union of those left and right sets,
%
%$L_i = \bigcup_\ell L_i(\mathcal{P}^\ell) \cup  L_i(\cT)$ and $R_i = \bigcup_\ell R_i(\mathcal{P}^\ell) \cup R_i(\cT).$
%
%%%%%%%%%%% Use below to save space %%%%%%%%
\[L_i = \bigcup_\ell L_i(\mathcal{P}^\ell) \cup  L_i(\cT)
~~~~~\mbox{and}~~~~~
R_i = \bigcup_\ell R_i(\mathcal{P}^\ell) \cup R_i(\cT)\,.\]
For $i=0$, we modify the definition and set $L_0 = V\setminus \{r\}$ and $R_0 = V\setminus \{s\}\,.$ See Figure~\ref{fig:graph}.

Let $\mathcal{A}$ be any {\em deterministic} distributed algorithm run on graph $\graph$ for computing a function $f$. Fix any input strings $x$ and $y$ given to $s$ and $r$ respectively. Let $\varphi_\cA(x, y)$ denote the execution of $\mathcal{A}$ on $x$ and $y$. Denote the {\em state} of the vertex $v$ at the end of round $t$ during the execution $\varphi_\cA(x, y)$ by $\sigma_\cA(v, t, x, y)$.

We note the following important property of distributed algorithms.
%
%\begin{fact}
{\em The state of a vertex $v$ at the end of time $t$ is uniquely determined by its input and the sequence of messages on each of its incoming links from time $1$ to $t$.}
%\end{fact}
%
Intuitively, this is because a distributed algorithm is simply a set of algorithms run on different nodes in a network. The algorithm on each node behaves according to its input and the sequence of messages sent to it so far.
From this, for example, we can conclude that in two different executions $\varphi_\cA(x, y)$ and $\varphi_\cA(x', y')$, a vertex reaches the same state at time $t$ (i.e., $\sigma_\cA(v, t, x, y)=\sigma_\cA(v, t, x', y')$) if and only if it receives the same sequence of messages on each of its incoming links.

%\begin{wrapfigure}{r}{0.5\linewidth}
\begin{figure}%[tpb]
  % Requires \usepackage{graphicx}
  \centering
  \scriptsize
    {
    \psfrag{A}[c]{$\cP^1$}
    \psfrag{B}[c]{$\cP^\ell$}
    \psfrag{C}[c]{$\cP^{\Gamma}$}
    \psfrag{a}[c]{$v_0^1$}
    \psfrag{b}[c]{$v_1^1$}
    \psfrag{c}[c]{$v_2^1$}
    \psfrag{d}{$v_{d^p-1}^1$}
    \psfrag{e}[c]{$v_0^\ell$}
    \psfrag{f}[c]{$v_1^\ell$}
    \psfrag{g}[c]{$v_2^\ell$}
    \psfrag{h}{$v_{d^p-1}^\ell$}
    \psfrag{i}[c]{$v_0^{\Gamma}$}
    \psfrag{j}[c]{$v_1^{\Gamma}$}
    \psfrag{k}[c]{$v_2^{\Gamma}$}
    \psfrag{l}{$v_{d^p-1}^{\Gamma}$}
    \psfrag{s}[r]{$s=u^p_0$}
    \psfrag{r}{$u^p_{d^p-1}=r$}
    \psfrag{m}[c]{$u^p_{1}$}
    \psfrag{n}[c]{$u^p_{2}$}
    \psfrag{o}{$u^0_0$}
    \psfrag{p}{$u^{p-1}_0$}
    \psfrag{q}{$u^{p-1}_1$}
    \psfrag{D}{$R_2$}
    \psfrag{E}{$R_1$}
    \psfrag{F}{$R_0$}
    %\subfigure[]{
%    \vspace{4mm}
    \includegraphics[width=0.4\linewidth]{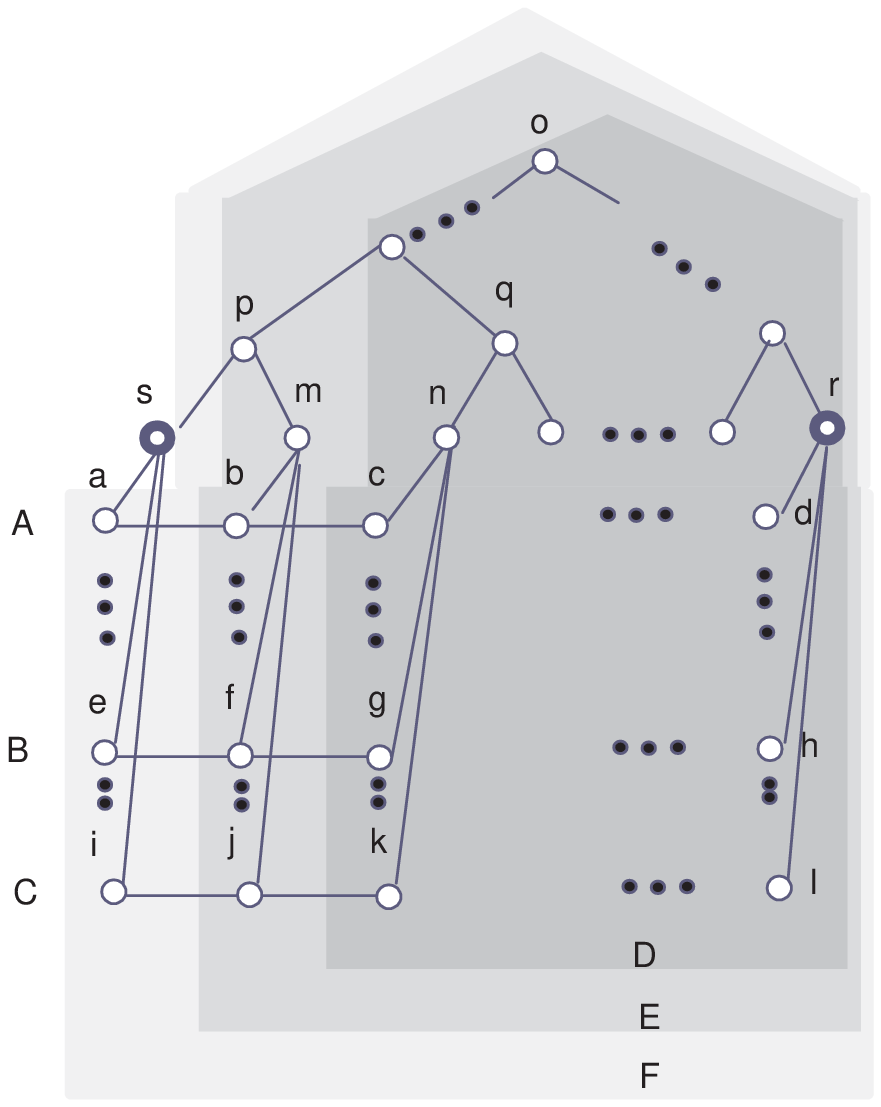}\label{fig:graph}\label{fig:ElkinGraph_RightLeftSet}
    %}
    }
  \caption{Examples of $i$-right sets.}
%  \vspace{4mm}
\end{figure}
%\end{wrapfigure}

For a given set of vertices $U=\{v_1, \ldots, v_\ell\}\subseteq V$, a {\em configuration}
%
%$C_\cA(U, t, x, y)$ = $<\sigma_\cA(v_1, t, x, y)$, $\ldots$, $\sigma_\cA(v_\ell, t, x, y)>$
%%%%%% Use above to save space %%%%%%%%%%%
\[C_\cA(U, t, x, y) = \langle\sigma_\cA(v_1, t, x, y), \ldots, \sigma_\cA(v_\ell, t, x, y)\rangle\]
is a vector of the states of the vertices of $U$ at the end of round $t$ of the execution $\varphi_\cA(x, y)$.
%i.e., $C_\cA(U, t, x, y)$ = $<\sigma_\cA(v_1, t, x, y), \ldots, \sigma_\cA(v_\ell, t, %x, y)>$
%
%
%

%\subsection{Proof of Theorem~\ref{thm:cc_to_distributed}}
\subsection{Observations}
We note the following crucial observations developed in \cite{PelegR00,Elkin06,LotkerPP06,KorKP10}.
%We provide them here with proofs for completeness.
We will need Lemma~\ref{thm:computing_subconfig} to prove Theorem~\ref{thm:cc_to_distributed} in the next subsection.

\begin{observation}\label{obs:crucial_observation}
For any set $U\subseteq U'\subseteq V$, $C_\cA(U, t, x, y)$ can be uniquely determined by $C_\cA(U', t-1, x, y)$ and all messages sent to $U$ from $V\setminus U'$ at time $t$.
\end{observation}
\begin{proof}
Recall that the state of each vertex $v$ in $U$ can be uniquely determined by its state $\sigma_\cA(v, t-1, x, y)$ at time $t-1$ and the messages sent to it at time $t$. Moreover, the messages sent to $v$ from vertices inside $U'$ can be determined by $C_\cA(U', t-1, x, y)$. Thus if the messages sent from vertices in $V\setminus U'$ are given then we can determine all messages sent to $U$ at time $t$ and thus we can determine $C_\cA(U, t, x, y)$.
\end{proof}

From now on, to simplify notations, when $\cA$, $x$ and $y$ are clear from the context, we use $C_{L_t}$ and $C_{R_t}$ to denote $C_\cA(L_t, t, x, y)$ and $C_\cA(R_t, t, x, y)$, respectively.
The lemma below states that $C_{L_t}$ ($C_{R_t}$, respectively) can be determined by $C_{L_{t-1}}$ ($C_{R_{t-1}}$, respectively) and $dp$ messages generated by some vertices in $R_{t-1}$ ($L_{t-1}$ respectively) at time $t$.
It essentially follows from Observation~\ref{obs:crucial_observation} and an observation that there are at most $d^p$ edges linking between vertices in $V\setminus R_{t-1}$ ($V\setminus L_{t-1}$ respectively) and vertices in $R_t$ ($L_t$ respectively).

\begin{lemma}\label{thm:computing_subconfig}
Fix any deterministic algorithm $\cA$ and input strings $x$ and $y$. For any $0<t<(d^p-1)/2$, there exist functions $g_L$ and $g_R$, $B$-bit messages $M^{L_{t-1}}_1, \ldots, M^{L_{t-1}}_{dp}$ sent by some vertices in $L_{t-1}$ at time $t$, and $B$-bit messages $M^{R_{t-1}}_1$, $\ldots$, $M^{R_{t-1}}_{dp}$  sent by some vertices in $R_{t-1}$ at time $t$ such that
\begin{align}
C_{L_t} &=g_L(C_{L_{t-1}}, M^{R_{t-1}}_1, \ldots, M^{R_{t-1}}_{dp}) \mbox{, and}\label{eq:left}\\
%~~~~~\mbox{and}~~~~
C_{R_t} &=g_R(C_{R_{t-1}}, M^{L_{t-1}}_1, \ldots, M^{L_{t-1}}_{dp})\label{eq:right}\,.
\end{align}
\end{lemma}

\begin{proof}
We prove Eq.~\eqref{eq:right} only. (Eq.~\eqref{eq:left} is proved in exactly the same way.) First, observe the following facts about neighbors of nodes in $R_t$.

\begin{itemize}
\item
All neighbors of all path vertices in $R_t$ are in $R_{t-1}$. {\em Example:} In Figure~\ref{fig:ElkinGraph_RightLeftSet}, path vertices in $R_2$ are $v^\ell_2, \ldots, v^\ell_{d^p-1}$ for $\ell=1, \ldots, \Gamma$. Observe that all neighbors of these vertices, i.e. $v^\ell_1, \ldots, v^\ell_{d^p-1}$ for all $\ell$ and $u^p_1, \ldots, u^p_{d^p-1}$, are in $R_1$.

\item All neighbors of all leaf vertices in $V(\cT)\cap R_t$ are in $R_{t-1}$. {\em Example:} In Figure~\ref{fig:ElkinGraph_RightLeftSet} vertices in $R_2$ are $u_2^p, \ldots, u_{d^p-1}^p$. Their neighbors, i.e. $v^\ell_2, \ldots, v^\ell_{d^p-1}$ for all $\ell$ and $u_2^{p-1}, \ldots, u^{p-1}_{d^{p-1}-1}$, are all in $R_1$.

\item For any non-leaf tree vertex $u^\ell_i$, for any $\ell$ and $i$, if $u^\ell_i$ is in $R_{t}$ then its parent and vertices $u^\ell_{i+1}, u^\ell_{i+2}, \ldots, u^\ell_{d^\ell-1}$ are in $R_{t-1}$.
    %Moreover, all neighbors of $u^\ell_{i+1}, u^\ell_{i+2}, \ldots, u^\ell_{d^\ell-1}$ are in $R_{t-1}$.
    {\em Example:} In Figure~\ref{fig:ElkinGraph_RightLeftSet}, $u^{p-1}_1$ is in $R_2$. Thus, its parent ($u^{p-2}_0$)
    %, right child ($u^p_3$),
    and $u^{p-1}_2, \ldots, u^{p-1}_{d^{p-1}-1}$ are in $R_1$.

    %The above observations imply that the only vertices in $R_t$ that have some neighbors outside $R_t$ are the `leftmost'' non-leaf vertex at level $\ell=0, \ldots, p-1$ (i.e., $u^\ell_i\in R_t$ such that $u^\ell_{i-1}\notin R_t$).

\item For any $i$ and $\ell$, if $u^\ell_{i}$ is in $R_t$ then all children of $u^\ell_{i+1}$ are in $R_t$ (otherwise, all children of $u^\ell_{i}$ are not in $R_t$ and so is $u^\ell_{i}$, a contradiction). {\em Example:}  In Figure~\ref{fig:ElkinGraph_RightLeftSet}, $u^{p-1}_1$ is in $R_2$. Thus, all children of $u^{p-1}_2$ are in $R_2$.
\end{itemize}

Let $u^\ell(R_{t})$ denote the leftmost vertex that is at level $\ell$ of $\cT$ and in $R_t$, i.e., $u^\ell(R_t)=u^\ell_i$ where $i$ is such that $u^\ell_i\in R_t$ and $u^\ell_{i-1}\notin R_t$. (For example, in Figure~\ref{fig:ElkinGraph_RightLeftSet}, $u^{p-1}(R_1)=u_0^{p-1}$ and $u^{p-1}(R_2)=u_1^{p-1}$.)
From the above observations, we conclude that the only neighbors of nodes in $R_t$ that are not in $R_{t-1}$ are children of $u^\ell(R_t)$, for all $\ell$. In other words, all edges linking between vertices in $R_t$ and $V\setminus R_{t-1}$ are in the following form: $(u^\ell(R_t), u')$ for some $\ell$ and child $u'$ of $u^\ell(R_t)$.

%Moreover, their only neighbors outside $R_t$ are their left children. For any $\ell< p$ and $t$, let $u^\ell(R_{t})$ denote the leftmost vertex that is at level $\ell$ of $\cT$ and in $R_t$, i.e., $u^\ell(R_t)=u^\ell_i$ where $i$ is such that $u^\ell_i\in R_t$ and $u^\ell_{i-1}\notin R_t$.  (For example, in Figure~\ref{fig:ElkinGraph_RightLeftSet}, $u^{p-1}(R_1)=u_0^{p-1}$ and $u^{p-1}(R_2)=u_1^{p-1}$.)

Setting $U'=R_{t-1}$ and $U=R_t$ in Observation~\ref{obs:crucial_observation}, we have that $C_{R_t}$ can be uniquely determined by $C_{R_{t-1}}$ and messages sent to $u^\ell(R_t)$ from its children in $V\setminus R_{t-1}$. Note that each of these messages contains at most $B$ bits since they correspond to a message sent on an edge in one round.

Observe further that, for any $t<(d^p-1)/2$, $V\setminus R_{t-1}\subseteq L_{t-1}$ since $L_{t-1}$ and $R_{t-1}$ share some path vertices. Moreover, each $u^\ell(R_t)$ has $d$ children. Therefore, if we let  $M^{L_{t-1}}_1, \ldots, M^{L_{t-1}}_{dp}$ be the messages sent from children of $u^0(R_t), u^1(R_t), \ldots, u^{p-1}(R_t)$ in $V\setminus R_{t-1}$ to their parents (note that if there are less than $dp$ such messages then we add some empty messages) then we can uniquely determine $C_{R_t}$ by $C_{R_{t-1}}$ and  $M^{L_{t-1}}_1, \ldots, M^{L_{t-1}}_{dp}$.  Eq.~\eqref{eq:right} thus follows.
\end{proof}

Using the above lemma, we can now prove Theorem~\ref{thm:cc_to_distributed}.

%\begin{proof}[Proof of Theorem~\ref{thm:cc_to_distributed}]
%
%\paragraph{The reduction:}
%
\subsection{Proof of the Simulation Theorem (cf. Theorem~\ref{thm:cc_to_distributed})}
Let $f$ be the function in the theorem statement. Let $\mathcal{A}_\epsilon$ be any $\epsilon$-error distributed algorithm for computing $f$ on network $\ElkinGraph$. Fix a random string $\bar{r}$ used by $\mathcal{A}_\epsilon$ (shared by all vertices in $\ElkinGraph$) and consider the {\em deterministic} algorithm $\mathcal{A}$ run on the input of $\mathcal{A}_\epsilon$ and the fixed random string $\bar{r}$.
Let $T_{\mathcal{A}}$ be the worst case running time of algorithm $\mathcal{A}$ (over all inputs).
We note that $T_\mathcal{A}<(d^p-1)/2$, as assumed in the theorem statement.
We show that Alice and Bob, when given $\bar{r}$ as the public random string, can simulate $\mathcal{A}$ using at most $2dpBT_\mathcal{A}$ communication bits, as follows.

Alice and Bob make $T_{\mathcal{A}}$ {\em iterations} of communications.
Initially, Alice computes $C_{L_0}$ which depends only on $x$. Bob also computes $C_{R_0}$ which depends only on $y$. In each iteration $t>0$, we assume that Alice and Bob know $C_{L_{t-1}}$ and $C_{R_{t-1}}$, respectively, before the iteration starts. Then, Alice and Bob will exchange at most $2dpB$ bits so that Alice and Bob know $C_{L_{t}}$ and $C_{R_{t}}$, respectively, at the end of the iteration.

To do this, Alice sends to Bob the messages $M^{L_{t-1}}_1, \ldots, M^{L_{t-1}}_{dp}$ as in Lemma~\ref{thm:computing_subconfig}. Alice can generate these messages since she knows $C_{L_{t-1}}$ (by assumption). Then, Bob can compute $C_{R_t}$ using Eq.~\eqref{eq:right} in Lemma~\ref{thm:computing_subconfig}. Similarly, Bob sends $dp$ messages to Alice and Alice can compute $C_{L_t}$. They exchange at most $2dpB$ bits in total in each iteration since there are $2dp$ messages, each of $B$ bits, exchanged.

After $T_\mathcal{A}$ iterations, Alice knows $C(L_{T_\mathcal{A}}, T_\mathcal{A}, x, y)$ and Bob knows $C(R_{T_\mathcal{A}}, T_\mathcal{A}, x, y)$. In particular, they know the output of $\mathcal{A}$ (output by $s$ and $r$) since Alice and Bob know the states of $s$ and $r$, respectively, after $\cA$ terminates. They can thus output the output of $\mathcal{A}$.

Since Alice and Bob output exactly the output of $A$, they will answer correctly if and only if $\cA$ answers correctly. Thus, if $\cA$ is $\epsilon$-error then so is the above communication protocol between Alice and Bob.
Moreover, Alice and Bob communicate at most $2dpBT_\mathcal{A}$ bits. The theorem follows.
%\end{proof}
 % Reduction from computing f on Elkin's graphs to 2-party communication complexity
\section{Distributed Verification of Set Disjointness and Equality Functions}\label{sec:function_verification}

In this section, we show lower bounds of distributed algorithms for verifying set disjointness and equality. The definitions of both problems can be found in Section~\ref{subsec:compcomp} and the model of distributed verification of functions can be found in Section~\ref{subsec:func_verification_def}.  The results in this section are simple corollaries of the Simulation Theorem (cf. Theorem~\ref{thm:cc_to_distributed}) and will serve as important building blocks in showing lower bounds in later sections.

\subsection{Randomized Lower Bound of Set Disjointness Function}
To prove the lower bound of verifying $\DISJfunc$, we simply use the communication complexity lower bound of computing $\DISJfunc$ \cite{BabaiFS86,KalyanasundaramS92,Bar-YossefJKS04,Razborov92}, i.e., $R^{cc-pub}_\epsilon(\DISJfunc)=\Omega(b)$ where $b$ is the size of input strings $x$ and $y$.

\begin{lemma}\label{lem:disj}
For any $\Gamma, d, p$, there exists a constant $\epsilon>0$ such that
$$R_\epsilon^{\ElkinGraph}(\DISJfunc)=\Omega(\min(d^p, \frac{b}{dpB})),$$
where $b$ is the size of input strings $x$ and $y$ of $\DISJfunc$; i.e., any $\epsilon$-error algorithm computing function $\DISJfunc$ on $\ElkinGraph$ requires $\Omega(\min(d^p, \frac{b}{dpB}))$ time.
\end{lemma}
\begin{proof}
If $R_\epsilon^{\ElkinGraph}(\DISJfunc)\geq (d^p-1)/2$
then $R_\epsilon^{\ElkinGraph}(\DISJfunc)=\Omega(d^p)$ and we are done. Otherwise the conditions of Theorem~\ref{thm:cc_to_distributed} are fulfilled and it implies that $R^{cc-pub}_\epsilon(\DISJfunc)\leq 2dpB\cdot R_\epsilon^{\ElkinGraph}(\DISJfunc)$. Now we use the fact that $R^{cc-pub}_\epsilon(\DISJfunc)$ $=\Omega(b)$ for the function $\DISJfunc$ on $b$-bit inputs, for some $\epsilon>0$ \cite{BabaiFS86,KalyanasundaramS92,Bar-YossefJKS04,Razborov92} (also see \cite[Example 3.22]{KNbook} and references therein). It follows that $R_\epsilon^{\ElkinGraph}(\DISJfunc)=\Omega(b/(dpB))$.
\end{proof}

\subsection{Deterministic Lower Bound of Equality Function}
To prove the lower bound of verifying $\EQfunc$, we simply use the deterministic communication complexity lower bound of computing $\EQfunc$ \cite{Yao79-cc}, i.e., $R^{cc-pub}_0(\EQfunc)=\Omega(b)$ where $b$ is the size of input strings $x$ and $y$ (see, e.g., \cite[Example 1.21]{KNbook} and references therein).

\begin{lemma}\label{lem:eq}
For any $\Gamma, d, p$,
$$R_0^{\ElkinGraph}(\EQfunc)=\Omega(\min(d^p, \frac{b}{dpB})),$$
where $b$ is the size of input strings $x$ and $y$ of $\EQfunc$; i.e., any deterministic algorithm computing function $\EQfunc$ on $\ElkinGraph$ requires $\Omega(\min(d^p, \frac{b}{dpB}))$ time.
\end{lemma}
\begin{proof}
If $R_0^{\ElkinGraph}(\EQfunc)\geq (d^p-1)/2$
then $R_0^{\ElkinGraph}(\EQfunc)=\Omega(d^p)$ and we are done. Otherwise, the conditions of Theorem~\ref{thm:cc_to_distributed} are fulfilled and it implies that $R^{cc-pub}_0(\EQfunc)\leq 2dpB\cdot R_0^{\ElkinGraph}(\EQfunc)$. Now we use the fact that $R^{cc-pub}_0(\EQfunc)$ $=\Omega(b)$ for the function $\EQfunc$ on $b$-bit inputs. It follows that $R_\epsilon^{\ElkinGraph}(\EQfunc)=\Omega(b/(dpB))$.
\end{proof}

 % Lower bound of set disjointness and Hamiltonian cycle verification
\section{Randomized Lower Bounds for Distributed Verification}\label{sec:randomized_lb}

In this section, we present randomized lower bounds for many verification problems on graphs of various diameters, as shown in Figure~\ref{fig:lower bounds on various diameters}. These problems are defined in Section~\ref{subsec:verification_def}. The key ingredient is the lower bound of verifying the set disjointness function on distributed networks (cf. Lemma~\ref{lem:disj}).
The general theorem is as follows.
%
%In this problem, we want to verify whether $H$ is connected and spans all nodes of $G$, i.e., every node in $G$ is incident to some edge in $H$. Definitions of other problems and proofs of their lower bounds are in \appotherlb.

\begin{theorem}\label{thm:all_verification}
For any $p\geq 1$, $B\geq 1$, and $n\in \{2^{2p+1}pB, 3^{2p+1}pB, \ldots\}$, there exists a constant $\epsilon>0$ such that any $\epsilon$-error distributed algorithm for any of the following problems requires $\Omega((n/(pB))^{\frac{1}{2}-\frac{1}{2(2p+1)}})$ time on some $\Theta(n)$-vertex graph of diameter $2p+2$ in the $B$ model: Spanning connected subgraph, cycle containment, $e$-cycle containment, bipartiteness, $s$-$t$ connectivity, connectivity, cut, edge on all paths, $s$-$t$ cut and least-element list.
\end{theorem}

In particular, for graphs with diameter $D=4$, we get $\Omega((n/B)^{1/3})$ lower bound and for graphs with diameter $D=\log{n}$ we get $\Omega(\sqrt{n/(B\log n)})$. Similar analysis also leads to a $\Omega(\sqrt{n/B})$ lower bound for graphs of diameter $n^\delta$ for any $\delta>0$, and $\Omega((n/B)^{1/4})$ lower bound for graphs of diameter three using the same analysis as in \cite{Elkin06}.
We note again that the lower bound holds even in the public coin model where every vertex shares a random string.

\paragraph{\bf Organization} This section is organized as follows. In the first three subsections, we show lower bounds that need a reduction from the set disjointness problem (i.e., problems in the third column in Figure~\ref{fig:all_reductions}): spanning connected subgraph verification in Subsection~\ref{subsec:spanconn_veri}, $s$-$t$ connectivity verification in Subsection~\ref{subsec:stconn_veri} and cycle containment, $e$-cycle containment, and bipartiteness verification in Subsection~\ref{subsec:cycle_veri} (these problems are proved together as they use the same construction). The lower bounds on the remaining problems (connectivity, cut, edges on all paths, $s$-$t$ cut and least-element list verification) are in Subsection \ref{sec:rand reduction from other problems}.

\subsection{Lower Bound of Spanning Connected Subgraph Verification Problem}\label{subsec:spanconn_veri}
The lower bound of spanning connected subgraph verification essentially follows from the following lemma which says that an algorithm for solving spanning connected subgraph verification can be used to compute $\DISJfunc$ as well.

\begin{lemma}\label{lem:spanning_connected}
For any $\Gamma$, $d\geq 2$, $p$ and $\epsilon\geq 0$, if there exists an $\epsilon$-error distributed algorithm for the spanning connected subgraph verification problem on graph $\ElkinGraph$ then there exists an $\epsilon$-error algorithm for verifying $\DISJfunc$ (on $\Gamma$-bit inputs) on $\ElkinGraph$ that uses the same time complexity.
\end{lemma}
\begin{proof}
Consider an $\epsilon$-error algorithm $\cA$ for the spanning connected subgraph verification problem, and suppose
that we are given an instance of the set disjointness problem with $\Gamma$-bit input strings $x$ and $y$ (given to $s$ and $r$). We use $\cA$
to solve this instance of the set disjointness problem by constructing $H$ as follows.

First, we mark all path edges and tree edges as participating in $H$. All spoke edges are marked as not participating in subgraph  $H$, except those incident to $s$ and $r$ for which we do the following:
For each bit $x_i$, $1 \le i \le \Gamma$, vertex $s$ indicates that the spoke edge $(s, v_0^i)$ participates in $H$ if and only if $x_i = 0$. Similarly, for each bit $y_i$, $1 \le i \le \Gamma$, vertex $r$ indicates that the spoke edge $(r, v_{d^p-1}^i)$ participates in $H$ if and only if $y_i = 0$. (See Figure~\ref{fig:connected spanning}.)

%%%%%%%%%%%%%%%%%%%% FIGURE: Connected Spanning %%%%%%%%%%%%%%%%%%%%%%%

\begin{figure}
  % Requires \usepackage{graphicx}
  \centering
     {
    \psfrag{A}[c]{$\cP^1$}
    \psfrag{B}[c]{$\cP^{\Gamma-1}$}
    \psfrag{C}[c]{$\cP^{\Gamma}$}
    \psfrag{a}[c]{$v_0^1$}
    \psfrag{b}[c]{$v_1^1$}
    \psfrag{c}[c]{$v_2^1$}
    \psfrag{d}{$v_{d^p-1}^1$}
    \psfrag{e}[c]{$v_0^\ell$}
    \psfrag{f}[c]{$v_1^\ell$}
    \psfrag{g}[c]{$v_2^\ell$}
    \psfrag{h}{$v_{d^p-1}^\ell$}
    \psfrag{i}[c]{$v_0^{\Gamma}$}
    \psfrag{j}[c]{$v_1^{\Gamma}$}
    \psfrag{k}[c]{$v_2^{\Gamma}$}
    \psfrag{l}{$v_{d^p-1}^{\Gamma}$}
    \psfrag{s}[r]{$s=u^p_0$}
    \psfrag{r}{$u^p_{d^p-1}=r$}
    \psfrag{m}[c]{$u^p_{1}$}
    \psfrag{n}[c]{$u^p_{2}$}
    \psfrag{o}{$u^0_0$}
    \psfrag{p}{$u^{p-1}_0$}
    \psfrag{q}{$u^{p-1}_1$}
    \psfrag{D}{$R_2$}
    \psfrag{E}{$R_1$}
    \psfrag{F}{$R_0$}
    %\subfigure[]{
    \includegraphics[height=0.5\linewidth]{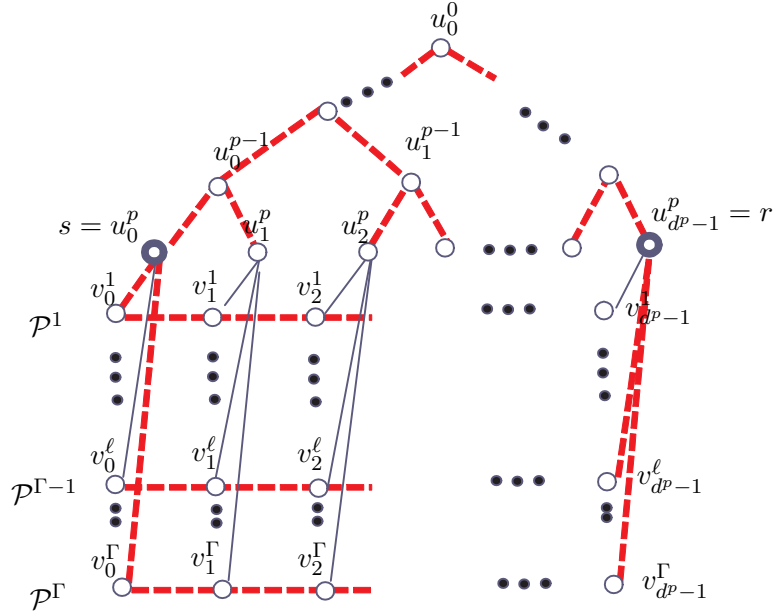}
    %}
   }
  \caption{Example of $H$ for the spanning connected subgraph problem (marked with dashed edges (red edges)) when $x=0...10 $ and $y=1...00$.}\vspace*{-\figspace ex}\label{fig:connected spanning}
\end{figure}

Note that the participation of all edges, except those incident to $s$ and $r$, is decided independently of
the input. Moreover, one round is sufficient for $s$ and $r$ to inform their neighbors of the participation of edges incident to them. Hence, one round is enough to construct $H$. Then, algorithm $\cA$ is started.

Once algorithm $\cA$ terminates, vertex $r$
determines its output for the set disjointness problem by stating that both
input strings are disjoint if and only if the spanning connected subgraph verification algorithm verified that
the given subgraph $H$ is indeed a spanning connected subgraph.

Observe that $H$ is a spanning connected subgraph if and only if for all $1\leq i\leq \Gamma$ at least one of the edges $(s,v_0^i)$ and $(r,v_{d^p-1}^i)$ is in $H$; thus, by the construction of $H$,
$H$ is a spanning connected subgraph if and only if the input strings $x, y$ are disjoint, i.e., for every $i$ either $x_i=0$ or $y_i=0$. Hence the resulting algorithm has correctly solved the given instance of the set disjointness  problem when $\cA$ correctly solve the spanning connected subgraph verification problem on the constructed subgraph $H$. This happens with probability at least $1-\epsilon$.
\end{proof}

Using Lemma~\ref{lem:disj}, we obtain the following result.

\begin{corollary}
For any $\Gamma, d, p$, there exists a constant $\epsilon>0$ such that any $\epsilon$-error algorithm for the spanning connected subgraph verification problem requires $\Omega(\min(d^p, \frac{\Gamma}{dpB}))$ time on some $\Theta(\Gamma d^p)$-vertex graph of diameter $2p+2$.
\end{corollary}

%\begin{corollary}
%For every $K\geq 2$, there exists a constant $\epsilon>0$ and a family of $n$-vertex graphs of diameter $O(Kn^{1/K})$ such that any $\epsilon$-error distributed algorithm for spanning connected subgraph verification in a $B$ model, for $B\geq 3$, needs $\Omega(n/(BK))$ time.
%\end{corollary}

In particular, if we consider $\Gamma=d^{p+1}pB$ then $\Omega(\min(d^p, \Gamma/(dpB))) = \Omega(d^p)$. Moreover, by  Lemma~\ref{lem:size}, $G(d^{p+1}pB, d, p)$ has $n=\Theta(d^{2p+1}pB)$ vertices and thus the lower bound of $\Omega(d^p)$ becomes $\Omega((n/(pB))^{\frac{1}{2}-\frac{1}{2(2p+1)}})$. Theorem~\ref{thm:all_verification} (for the case of spanning connected subgraph) follows.

\subsection{Lower Bound of $s$-$t$ Connectivity Verification Problem}\label{sec:st_lowerbound}\label{subsec:stconn_veri}
We again modify the proof of Lemma~\ref{lem:spanning_connected} to prove the following lemma.

%Similar to the lower bound of the spanning connected subgraph verification problem, the lower bounds of $s$-$t$ connectivity follow from the following lemma.

\begin{lemma}\label{lem:stcon}
For any $\Gamma$, $d\geq 2$, $p$ and $\epsilon\geq 0$ if there exists an $\epsilon$-error distributed algorithm for the $s$-$t$ connectivity verification problem on graph $\graph$ then there exists an $\epsilon$-error algorithm for verifying $\DISJfunc$ (on $\Gamma$-bit inputs) on $\ElkinGraph$ that uses the same time complexity.
\end{lemma}
\begin{proof}
We use the same argument as in the proof of Lemma~\ref{lem:spanning_connected} except that we construct the subgraph $H$ as follows.

First, all path edges are marked as participating in subgraph $H$. All tree edges are marked as not participating in $H$. All spoke edges, except those incident to $s$ and $r$, are also marked as not participating. %
For each bit $x_i$, $1 \le i \le \Gamma$, vertex $s$ indicates that the spoke edge $(s, v_0^i)$ participates in $H$ if and only if $x_i = 1$. Similarly, for each bit $y_i$, $1 \le i \le \Gamma$, vertex $r$ indicates that the spoke edge $(r, v_{d^p-1}^i)$ participates in $H$ if and only if $y_i = 1$. (See Figure~\ref{fig:stcon}.)

%%%%%%% FIGURE: Example of H for s-t connectivity %%%%%%%%%%%%%%%

\begin{figure}
\centering
   {
   \footnotesize
    \psfrag{A}[c]{$\cP^1$}
    \psfrag{B}[c]{$\cP^{\Gamma-1}$}
    \psfrag{C}[c]{$\cP^{\Gamma}$}
    \psfrag{a}[c]{$v_0^1$}
    \psfrag{b}[c]{$v_1^1$}
    \psfrag{c}[c]{$v_2^1$}
    \psfrag{d}{$v_{d^p-1}^1$}
    \psfrag{e}[c]{$v_0^\ell$}
    \psfrag{f}[c]{$v_1^\ell$}
    \psfrag{g}[c]{$v_2^\ell$}
    \psfrag{h}{$v_{d^p-1}^\ell$}
    \psfrag{i}[c]{$v_0^{\Gamma}$}
    \psfrag{j}[c]{$v_1^{\Gamma}$}
    \psfrag{k}[c]{$v_2^{\Gamma}$}
    \psfrag{l}{$v_{d^p-1}^{\Gamma}$}
    \psfrag{s}[c]{$s=u^p_0$}
    \psfrag{r}[c]{$u^p_{d^p-1}=r$}
    \psfrag{m}[c]{$u^p_{1}$}
    \psfrag{n}[c]{$u^p_{2}$}
    \psfrag{o}{$u^0_0$}
    \psfrag{p}{$u^{p-1}_0$}
    \psfrag{q}{$u^{p-1}_1$}
    \psfrag{D}{$R_2$}
    \psfrag{E}{$R_1$}
    \psfrag{F}{$R_0$}
    \includegraphics[height=0.45\linewidth]{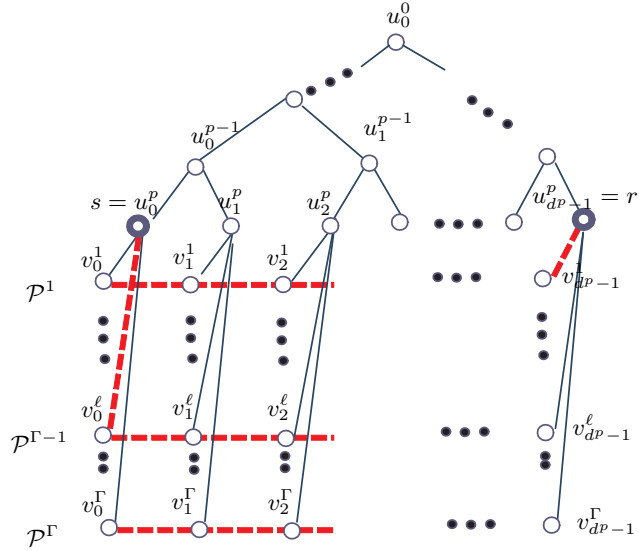}
   }
  \caption{\footnotesize Example of $H$ for $s$-$t$ connectivity problem (marked with dashed edges (red edges)) when $x=0...10 $ and $y=1...00$.}\label{fig:stcon} %\vspace*{-\figspace ex}
\end{figure}

%Once algorithm $\cA_{st}$ terminates, vertex $r$ determines its output for the set disjointness problem by stating that both input strings are disjoint if and only if $s$-$r$ connectivity verification algorithm verified that $s$ and $r$ are {\em not} connected in the given subgraph.

%For the correctness of this algorithm, o
Observe that $s$ and $r$ are connected in $H$ if and only if there exists $1\le i\le \Gamma$ such that both edges $(v_0^i,s), (v_{d^p-1}^i,r)$ are in $H$; thus, by the construction of $H$, $H$ is $s$-$r$ connected if and only if the input strings $x$ and $y$ are {\em not} disjoint.
%Hence the resulting algorithm has correctly solved the given instance of the set disjointness problem.
\end{proof}

\subsection{Lower Bounds of Cycle Containment, $e$-Cycle Containment, and Bipartiteness Verification Problems}\label{sec:reduction from disj}\label{subsec:cycle_veri}
We modify the proof of Lemma~\ref{lem:stcon} to prove the following lemma which says that an algorithm for solving problems in this section can be used to compute $\DISJfunc$.

%A deterministic lower bound follows due to Theorem \ref{thm:ham}. Now extend it to hold for randomized algorithms as well:
\begin{lemma}\label{lem:cycle ecycle bipartiteness}
For any $\Gamma$, $d\geq 2$, $p$ and $\epsilon\geq 0$ if there exists an $\epsilon$-error distributed algorithm for solving either the cycle containment, $e$-cycle containment or bipartiteness verification problem on graph $\graph$ then there exists an $\epsilon$-error algorithm for verifying $\DISJfunc$ (on $\Gamma$-bit inputs) on $\graph$ that uses the same time complexity.
\end{lemma}
\begin{proof} We prove this lemma by modifying the proof of Lemma~\ref{lem:stcon}. We only note the key difference here.

%\emph{cycle verification problem:} We modify the proof of Lemma~\ref{lem:stcon} as follows. Given input vectors $x^s$ and $x^r$ of length $m^K$ we construct $H$ in the same way as in the proof of Lemma \ref{lem:stcon}, {\em except} that highway edges in $\cH^1$ are marked as participating in $H$.
%
\emph{Cycle containment verification problem:}
We construct $H$ in the same way as in the proof of Lemma~\ref{lem:stcon} {\em except} that the tree edges are participating in $H$ (see Figure~\ref{fig:cycle}).

%%%%%%%% FIGURE: Example of H for cycle containment %%%%%%%%%%%%%%%

\begin{figure}
\centering
   {
   \footnotesize
    \psfrag{A}[c]{$\cP^1$}
    \psfrag{B}[c]{$\cP^{\Gamma-1}$}
    \psfrag{C}[c]{$\cP^{\Gamma}$}
    \psfrag{a}[c]{$v_0^1$}
    \psfrag{b}[c]{$v_1^1$}
    \psfrag{c}[c]{$v_2^1$}
    \psfrag{d}{$v_{d^p-1}^1$}
    \psfrag{e}[c]{$v_0^\ell$}
    \psfrag{f}[c]{$v_1^\ell$}
    \psfrag{g}[c]{$v_2^\ell$}
    \psfrag{h}{$v_{d^p-1}^\ell$}
    \psfrag{i}[c]{$v_0^{\Gamma}$}
    \psfrag{j}[c]{$v_1^{\Gamma}$}
    \psfrag{k}[c]{$v_2^{\Gamma}$}
    \psfrag{l}{$v_{d^p-1}^{\Gamma}$}
    \psfrag{s}[c]{$s=u^p_0$}
    \psfrag{r}[c]{$u^p_{d^p-1}=r$}
    \psfrag{m}[c]{$u^p_{1}$}
    \psfrag{n}[c]{$u^p_{2}$}
    \psfrag{o}{$u^0_0$}
    \psfrag{p}{$u^{p-1}_0$}
    \psfrag{q}{$u^{p-1}_1$}
    \psfrag{D}{$R_2$}
    \psfrag{E}{$R_1$}
    \psfrag{F}{$R_0$}
    \includegraphics[height=0.45\linewidth]{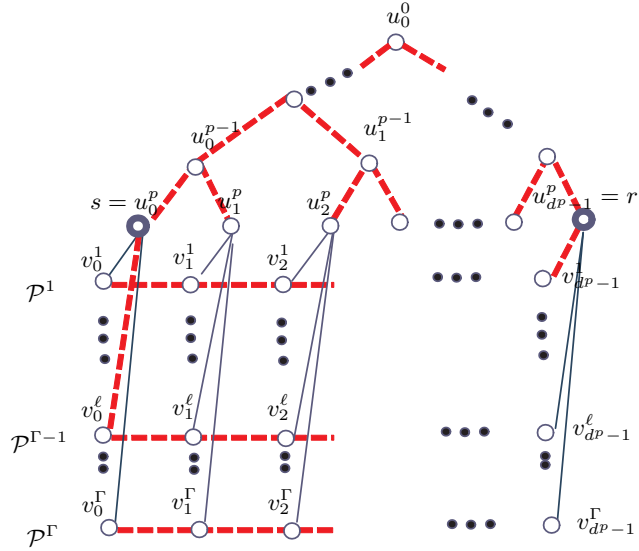}
   }
  \caption{\footnotesize Example of $H$ for the cycle and $e$-cycle containment and bipartiteness verification problem when $x=0...10 $ and $y=1...00$.}\label{fig:cycle} %\vspace*{-\figspace ex}
\end{figure}
%{
%\psfrag{A}{$\cH^2$}
%
%\psfrag{B}{$\cP^1$}
%
%\psfrag{C}{$\cP^2$}
%
%\psfrag{D}{$\cP^{m^3}$}
%
%\psfrag{E}{$\cH^1$}
%
%\psfrag{s}{$s$}
%
%\psfrag{r}{$r$}
%
%\psfrag{e}{$e$}
%
%\begin{figure}
%  % Requires \usepackage{graphicx}
%  \centering
%  \includegraphics[width=0.55\linewidth]{cycle.eps}\\
%  \caption{Example of $H$ for cycle and $e$-cycle containment verification problem (marked with thick red edges) when $x=01...0 $ and $y=10...0$.}\label{fig:cycle}
%\end{figure}
%}
%%%%%%%%%%%% END OF FIGURE %%%%%%%%%%%%%%

In the case that the input strings are disjoint, $H$ will consist of the tree connecting $s$ and $r$ as well as 1) paths connected to $s$ but not to $r$, 2) paths connected to $r$ but not to $s$ and 3) paths connected neither to $r$ nor $s$. Thus there is no cycle in $H$.
In the case that the input strings are not disjoint, we let $i$ be an index that makes them not disjoint, that is $x_i=y_i=1$. This causes a cycle in $H$ consisting of some tree edges and path $\cP^i$ that are connected by edges $(s,v_0^i)$ and $(v_{d^p-1}^i,r)$ at their endpoints. Thus we have the following claim.
\begin{claim}
$H$ contains a cycle if and only if the input strings are not disjoint.
\end{claim}

\emph{$e$-cycle containment verification problem:} We use the previous construction for $H$ and let $e$ be the tree edge adjacent to $s$ (i.e., $e$ connects $s$ to its parent). Observe that, in this construction, $H$ contains a cycle if and only if $H$ contains a cycle containing $e$. Therefore, we have the following claim.
\begin{claim}\label{claim:cycle}
$e$ is contained in a cycle in $H$ if and only if the input strings are not disjoint.
\end{claim}

\emph{Bipartiteness verification problem:}
Finally, we can verify if such an edge $e$ is contained in a cycle by verifying the bipartiteness. First, we replace $e=(s,u^{p-1}_0)$ by a path $(s,v',u^{p-1}_0)$, where $v'$ is an additional/virtual vertex. This can be done without changing the input graph $G$ by having vertex $s$ simulated algorithms on both $s$ and $v'$. The communication between $s$ and $v'$ can be done internally. The communication between $v'$ and $u^{p-1}_0$ can be done by $s$. We construct $H'$ the same way as $H$ with both $(s,v')$ and $(v',u^{p-1}_0)$ marked as participating.

We observe that if the input strings are not disjoint, then either $H$ or $H'$ are not bipartite. To see this, consider two cases: when $d^p$ is even and odd. When $d^p$ is even and the input strings are not disjoint, there exists $i$ such that there is a cycle in $H$ consisting of some tree edges (including $e$) and path $P_i$ that are connected by edges $(s,v_0^i)$ and $(v_{d^p-1}^i,r)$ at their endpoints. This cycle is of length $2p+(d^p-1)+2$ -- an odd number causing $H$ to be not bipartite. If $d^p$ is odd, then by the same argument there is an odd cycle of length $(2p+1)+(d^p-1)+2$ in $H'$ (this cycle includes the edges $(s,v')$ and $(v',u^{p-1}_0)$ that replaces $e$); thus $H'$ is not bipartite.

Now we consider the converse: If the input strings are disjoint, then $H$ does not contain a cycle by the argument of the proof of the cycle containment problem (which uses the same graph). It follows that $H'$ does not contain a cycle as well. Therefore, we have the following claim.

%Now, if $H'$ is not bipartite then there is an odd cycle containing $(s,v',h^1_{m^{K-1}})$ which means that $H$ contains a cycle containing $e$. Conversely, if $H$ contains a cycle containing $e$ then $H'$ has an odd cycle containing $(s,v',h^1_{m^{K-1}})$ Therefore, we have the following claim.

%

\begin{claim}\label{claim:bipartite}
$H$ and $H'$ are both bipartite if and only if the input strings are disjoint.
\end{claim}
\end{proof}

We note that the above reduction for the bipartiteness verification problem might seem to suggest that one can also prove the lower bound of this problem by reducing from the $e$-cycle verification problem. However, this is not the case. The reason is that the above proof relies on the fact that $H$ and $H'$ each contains at most one cycle and such cycle must contain $e$. In general, this might not be the case.

\subsection{Lower Bounds of Connectivity, Cut, Edges on All Paths, $s$-$t$ Cut and Least-element List Verification Problems}\label{sec:rand reduction from other problems}
Lower bounds of verification problems in this section are proven using the lower bounds of problems in Section~\ref{subsec:spanconn_veri}, \ref{subsec:stconn_veri} and \ref{sec:reduction from disj}.

\paragraph{Connectivity verification problem}
We reduce from the spanning connected subgraph verification problem. Let $\mathcal A(G, H)$ be an algorithm that verifies if $H$ is connected in $O(\tau(n))$ time on any $n$-vertex graph $G$ and subgraph $H$. Now we will use this algorithm to verify whether a subgraph $H'$ of $G$ is connected or not.

Recall that, by definition,
$H'$ is a spanning connected subgraph if and only if every node is incident to at least one edge in $H'$ and $H'$ is connected.
Verifying that every node is incident to at least one edge in $H'$ can be done locally and all nodes can be notified if this is not the case in $O(D)$ rounds (by broadcasting). Checking if $H'$ is connected can be done in $O(\tau(n))$ rounds by running $\mathcal A(G, H')$. The total running time for checking if $H'$ is a spanning connected subgraph is thus $O(\tau(n)+D)$.
The lower bound of the spanning connected subgraph problem thus applies to the connectivity verification problem as well.

\iffalse \danupon{I removed this problem to make the paper slightly shorter}
%
\emph{$k$-component verification problem:} The above argument can be extended to show the lower bound of $k$-component problem, as follows. Suppose again that we want to check if $H$ is a spanning connected subgraph. Now we add $k-1$ virtual nodes adjacent to some node $s$ in $G$. These nodes are added to $H$ (denote the resulting subgraph by $H'$) but will not be incident to any edges in $H'$ and are simulated by $s$. Observe that the new graph, say $G'$, has diameter $D'=D+1$ and the number of nodes is $n'=n+k\leq 2n$. Moreover, $H$ is a spanning tree of $G$ if and only if $H'$ has $k$ connected component in $G'$ (the spanning subgraph $H$ of $G$ plus $k-1$ single nodes). Therefore, if we can check if $H$ has at most $k$ ($k$ constant) connected component in $G'$ in $O(\tau(n'))$ time then we can also check if $H$ is a spanning connected subgraph in $G$.
\fi

\paragraph{Cut verification problem} We again reduce from the spanning connected subgraph problem. Given a subgraph $H$, we verify if $H$ is a spanning connected subgraph as follows. Let $\bar{H}$ be the graph obtained by removing edges $E(H)$ of $H$ from $G$. Recall that $H$ is a spanning connected subgraph if and only if $\bar{H}$ is not a cut (see definition of a cut in Section~\ref{subsec:verification_def}). Thus, we verify if $\bar{H}$ is a cut and announce that $H$ is a spanning connected subgraph if and only if $\bar{H}$ is not a cut.

%\paragraph{Cut verification problem} We again reduce from the spanning connected subgraph problem. Recall that $H$ is a cut if and only if $G$ is {\em not} connected when we remove edges in $H$. In other words, $H$ is a cut if and only if $\bar{H}$ is not a spanning connected subgraph of $G$ where $\bar{H}$ is the graph resulting from removing edges in $H$.
%
%Thus, given a subgraph $H'$, we verify if $H'$ is a spanning connected subgraph as follows. Let $H''$ be the graph obtained by removing edges $E(H')$ of $H'$ from $G$. Recall again that  $H'$ is a spanning connected subgraph if and only if $H''$ is not a cut. Thus, we verify if $H''$ is a cut. We announce that $H'$ is a spanning connected subgraph if and only if $H''$ is verified not to be a cut.

\paragraph{$s$-$t$ cut verification problem}
We reduce from $s$-$t$ connectivity. Similar to above, we use the fact that $H$ is $s$-$t$ connected if and only if $\bar{H}$ is {\em not} an $s$-$t$ cut. %We now apply the lower bound of the $s$-$t$ connectivity verification problem.

\paragraph{Least-element list verification problem} We reduce from $s$-$t$ connectivity. We set the rank of $s$ to $0$ and the rank of other nodes to any distinct positive integers. We assign weight $0$ to all edges in $H$ and $1$ to other edges. Give a set $S=\{<s, 0>\}$ to vertex $t$. Then we verify if $S$ is the least-element list of $t$. Observe that if $s$ and $t$ are connected by $H$ then the distance between them must be $0$ and thus $S$ is the least-element list of $t$. Conversely, if $s$ and $t$ are not connected then the distance between them will be at least one and $S$ will not be the least-element list of $t$.

\paragraph{Edge on all paths verification problem}
We reduce from the $e$-cycle containment problem using the following observation: $H$ does not contain a cycle containing $e$ if and only if $e$ lies on all paths between $u$ and $v$ in $H$ where $u$ and $v$ are two nodes incident to $e$.

\section{Deterministic Lower Bounds of Distributed Verification}\label{sec:deterministic_SHORT}\label{sec:deterministic}

In this section, we present deterministic lower bounds for Hamiltonian cycle, spanning tree and simple path verification. These problems are defined in Section~\ref{subsec:verification_def}. These lower bounds are proved in almost the same way as in Section~\ref{sec:randomized_lb}. The only difference is that we reduce from the deterministic lower bound of the {\em Equality} problem (cf. Lemma~\ref{lem:eq}).

\begin{theorem}\label{thm:all_deterministic_lb}
For any $p$, $B\geq 1$, and $n$ $\in$ $\{2^{2p+1}pB$, $3^{2p+1}pB$, $\ldots\}$, any deterministic distributed algorithm for any of the following verification problems requires $\Omega((\frac{n}{pB})^{\frac{1}{2}-\frac{1}{2(2p+1)}})$ time on some $\Theta(n)$-vertex graph of diameter $2p+2$ in the $B$ model: Hamiltonian cycle, spanning tree, and simple path verification.
\end{theorem}

%We note that our lower bound of spanning tree verification simplifies and generalizes the lower bound of {\em minimum} spanning tree verification shown in \cite{KorKP10}.

We first prove the lower bound of the Hamiltonian cycle problem and later extend to other problems.

\subsection{Lower Bound of Hamiltonian Cycle Verification Problem}\label{sec:hamiltonian}
%\begin{lemma}
%For any $\Gamma$, $d\geq 2$ and $p$, if there exists a deterministic distributed algorithm for the Hamiltonian cycle verification problem on graph $\ElkinGraph$ in time $T$ then there exists a deterministic algorithm from verifying $\EQfunc$ (on $\Gamma$-bit inputs) on $\ElkinGraph$ that uses the same time complexity.
%\end{lemma}
%\begin{proof}
%
We construct $G(\Gamma,2,p)'$ from $G(\Gamma,2,p)$ by adding edges and vertices to $G(\Gamma,2,p)$. We argue that the Simulation Theorem (cf. Theorem~\ref{thm:cc_to_distributed}) holds on $G(\Gamma,2,p)'$ as well.
Note that since $d=2$, $\cT$ is a binary tree. Let $m$ be $d^p-1$.

First we add edges in such a way that the subgraph induced by the vertices $\{v_0^1,\dots,v_0^{\Gamma}\}$ is a clique and the subgraph induced by the vertices $\{v_{m}^1,\dots,v_{m}^{\Gamma}\}$ is a clique as well. Observe that the Simulation Theorem  holds on $G(\Gamma,2,p)$ with this change since it only relies on Lemma~\ref{thm:computing_subconfig} which is true on the modified graph as well.

Now we add edges $(u_i^p,u_{i+1}^p)$ for all $0\leq i \leq m-1$. Thus we shorten the distance between each pair of nodes $u_i^p$ and $u_{i+1}^p$ from at most three (using two spoke edges and one path-edge from the according ${\mathcal P}^l$) to one (red dashed edges in Figure \ref{fig:ham4}). This will affect the lower bound by at most a constant factor. Now for each node $u_i^l$ we add a path of length $p-l+1$ containing $p-l$ new nodes connecting $u_i^l$ to $u_{i\cdot d^{p-l}+1}^p$ (green dotted paths/nodes in Figure \ref{fig:ham4}). This will increase the number of messages needed in Lemma~\ref{thm:computing_subconfig} by a factor of two. Thus, the Simulation Theorem still holds.

%This will affect the lower bound by at most a constant factor as well since it will increase the amount of messages that can be sent through the graph within a certain time by at most $3$ times.

Further, we add the edges $(v_{m-1}^{\Gamma-2},u_0^p)$, $(v_{m}^{\Gamma-3},u_0^p)$ and $(v_m^\Gamma,u_0^0)$. This will again increase the number of messages needed in Lemma~\ref{thm:computing_subconfig} by a constant factor of two and thus the Simulation Theorem still holds.

Finally, we add the following three edges $(u_0^0,v_{m}^{\Gamma})$, $(s, v_{m-1}^{\Gamma-2})$, and $(s, v_{m-1}^{\Gamma-3})$. Adding these three edges will increase the number of messages needed in Lemma~\ref{thm:computing_subconfig} by at most three and thus the Simulation Theorem still holds. This completes the description of $G(\Gamma,2,p)'$.

\begin{figure}[htbp]
\begin{center}
\includegraphics[width=0.8\textwidth]{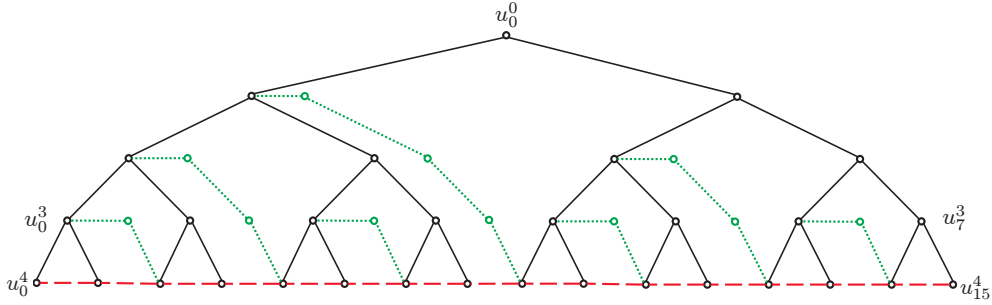}
\end{center}
\caption{Example of the modification of the tree-part of $G(\Gamma, 2, p)$ in the case $p=4$. The red dashed edges are new edges $(u_i^p,u_{i+1}^p)$ and the green dotted edges form new paths between nodes connecting $u_i^l$ to $u_{i\cdot d^{p-l}+1}^p$.}
\label{fig:ham4}
\end{figure}

To simplify and shorten the proof, we do some preparation. First, we consider strings $x$ and $y$ of length $b$ and define $\Gamma$ to be $ 2+12b$ -- this changes the bound only by a constant factor. Now, from $x$ and $y$, we construct strings of length $m$ (we assume $m$ to be even)
\begin{align*}
{x}{'}&:=1 x_1 01 x_1 01 x_2 01 x_2 01 \dots 01 x_{b} 01 x_{b} 01 \overline{x_1} 01 \overline{x_1} 01 \dots 01 \overline{x_{b}} 01 \overline{x_{b}} 010,\\
{y}{'}&:=1 y_1 01 y_1 01 y_2 01 y_2 01\dots 01 y_{b} 01 y_{b} 01 \overline{y_1} 01 \overline{y_1} 01 \dots 01 \overline{y_{b}} 01 \overline{y_{b}} 010
\end{align*}
%\[{x}{'}:=1 x_1 01 x_1 01 x_2 01 x_2 01 \dots 01 x_{b} 01 x_{b} 01 \overline{x_1} 01 \overline{x_1} 01 \dots 01 \overline{x_{b}} 01 \overline{x_{b}} 010,\]
%\[{y}{'}:=1 y_1 01 y_1 01 y_2 01 y_2 01\dots 01 y_{b} 01 y_{b} 01 \overline{y_1} 01 \overline{y_1} 01 \dots 01 \overline{y_{b}} 01 \overline{y_{b}} 010\]
where $\overline{x_i}$ and $\overline{y_i}$ denote negations of $x_i$ and $y_i$ respectively.

Now we construct $H$ in five stages. In the first stage we create some short paths that we call lines. In the next two stages we construct from these lines two paths $S_1$ and $S_2$ by connecting the lines in special ways with each other (the connections depend on the input strings). In the fourth stage we construct a path $S_3$ that will connect the leftover lines with each other.
These three paths, $S_1$, $S_2$, and $S_3$, will cover all nodes. The final stage is to connect all three paths together.

We will construct these paths in such a way that, the input strings are equal if and only if the resulting graph $H$ is a Hamiltonian cycle. Later we will observe that in the case the strings are equal all three paths will look like disjoint paths when using the graph layout of Figure \ref{fig:ElkinGraph}. The formal description of the five stages will be accompanied by a small example in Figures \ref{fig:ham1} and \ref{fig:ham2}. Here, $x=01=y$.

\paragraph{\bf Stage 1} We create the lines by marking most path edges (to be more precise, all edges $(v_j^i,v_{j+1}^i)$ for all $i\in [1,\Gamma]$ and $j\in\{2,\dots,m-2\}$ for $j\in\{1,\dots,m-1\}$) as participating in subgraph $H$. In addition we add the edges $(v_{m-1}^1,v_{m}^1)$ and $(v_{0}^{1},v_{1}^{1})$ to $H$. These basic elements are called lines now (see Figure \ref{fig:ham1}).

%Then nodes $s$ and $r$ construct $H$ within three rounds as follows:
\paragraph{\bf Stage 2} Define path $S_1$ as follows. All spoke edges incident to $\cH^1$ are marked as {\em not}
participating in $H$, except those incident to $s$ and $r$. For each bit $x'_i$, $1 \le i \le \Gamma$, vertex $s$ indicates that the edge $(v_0^i,v_0^{i+1})$ participates in $H$ if and only if $x'_i = 1$. Similarly, for each bit $y'_i$, $1 \le i \le \Gamma$, vertex $r$ indicates that the spoke edge $(v_{m}^i,v_{m}^{i+1})$ participates in $H$ if and only if $y'_i = 0$. Furthermore for $2\leq i\leq \Gamma$ each edge $(v_0^i,v_1^i)$ participates in $H$ if and only if $x'_{i-1}\neq x'_{i}$. Similarly for $2\leq i\leq \Gamma$ each edge $(v_{m-1}^i,v_{m}^i)$ participates in $H$ if and only if $y'_{i-1}\neq y'_{i}$. In addition we let edges $(v_0^1,v_1^1)$ and $(v_{m-1}^1,v_m^1)$ participate in $H$. We denote the path that results from connecting the lines according to the rules above by $S_1$. An example is given in Figure \ref{fig:ham1}.
\begin{figure}[htbp]
\begin{center}
\includegraphics[width=\textwidth]{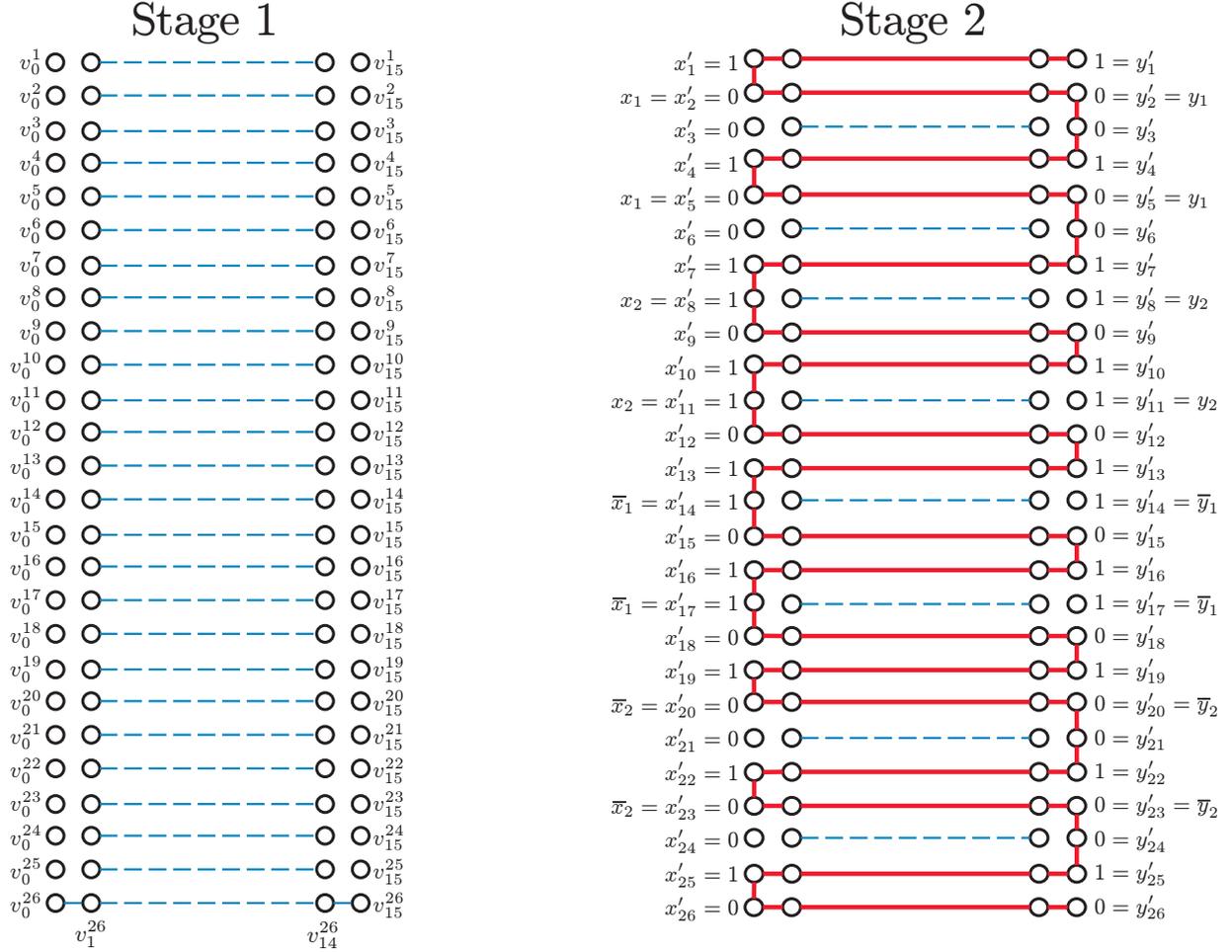}
\end{center}
\caption{Example of the reduction using input strings $x=01=y$. Note that $\Gamma$ is $12\cdot2+2=26$ and we use $d=2$ and $p=4$. In Stage 1, we add blue paths to $H$, that are displayed by blue dashed lines. In Stage 2, we create $S_1$, the red-colored path that looks like a snake.}\label{fig:ham1}
\end{figure}

\paragraph{\bf Stage 3} Define $S_2$ as follows. We connect the lines not covered in Stage 3 (except the tree $\cT$) and those nodes that are not covered by any path or line. In particular, on the left side of the graph, for $0\leq i\leq 2b$, we do the following.
\begin{itemize}
\item If ${x}'_{2+6i}=0$ (and thus ${x}'_{5+6i}=0$ due to the definition of $x'$) then edges $(v_1^{3+6i},v_0^{3+6i})$, $(v_0^{3+6i},v_0^{6+6i})$ and $(v_0^{6+6i},v_1^{6+6i})$ are indicated to participate in $H$.
\item If $x'_{2+6i}=1$ (and thus $x'_{5+6i}=1$ due to the definition of ${x}'$), edges $(v_{m}^{3+6i},v_{m-1}^{3+6i})$, $(v_1^{3+6i},v_1^{6+6i})$ and $(v_{m-1}^{6+6i},v_{m}^{6+6i})$ will participate in $H$.
\end{itemize}
On the right side of the graph, for $0\leq i\leq 2b$ we indicate the following edges to participate in $H$:
\begin{itemize}
\item $(v_{m}^{5+6i},v_{m}^{2+6(i+1)})$ if $y'_{5+6i}=0$ and $y'_{2+6(i+1)}=0$.
\item $(v_{m}^{5+6i},v_{m-1}^{3+6(i+1)})$ if $y'_{5+6i}=0$ and $y'_{2+6(i+1)}=1$.
\item $(v_{m-1}^{6+6i},v_{m}^{2+6(i+1)})$ if $y'_{5+6i}=1$ and $y'_{2+6(i+1)}=0$.
\item $(v_{m-1}^{6+6i},v_{m-1}^{3+6(i+1)})$ if $y'_{5+6i}=1$ and $y'_{2+6(i+1)}=1$.
\end{itemize}

We denote the path that results from connecting lines according to the rules above by $S_2$. An example is given in Figure \ref{fig:ham2}.

\begin{figure}[htbp]
\begin{center}
\includegraphics[width=0.8\textwidth]{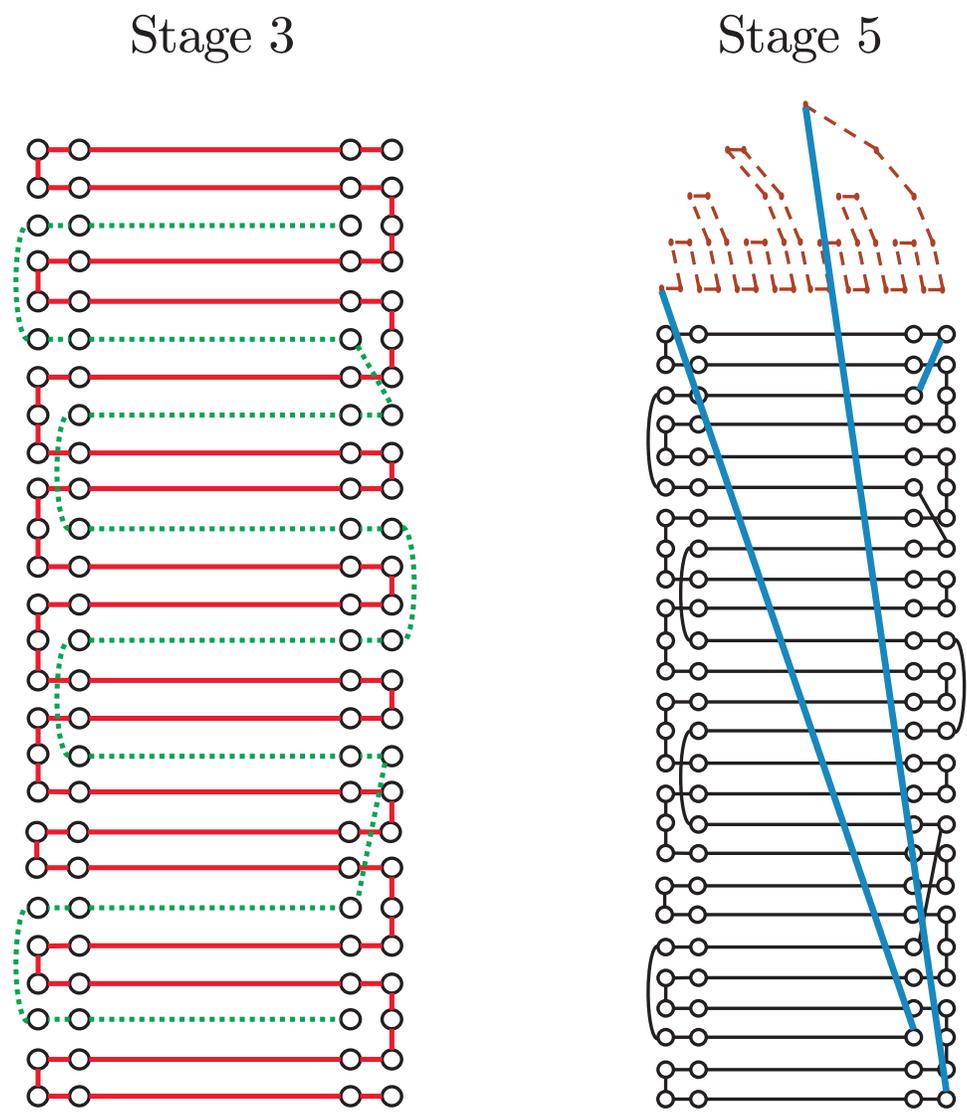}
\end{center}
\caption{Continuation of the example started in Figure \ref{fig:ham1}. In Stage 3, we add $S_2$ in dotted green (dotted lines). $S_3$ is displayed in dashed brown and added in stage four. Finally we connect $S_1$, $S_2$ and $S_3$ by bold blue edges to a Hamiltonian cycle in Stage 5.}\label{fig:ham2}
\end{figure}

\paragraph{\bf Stage 4} We include edges of the modified tree in a canonical way in $H$ to form a path $S_3$ that covers all nodes in $\cT$ and starts and ends at $u_0^p$ and $u_0^0$ respectively, as follows. For all odd $i$ such that $0\leq i\leq m-1$, we include the edge $(u_i^p,u_{i+1}^p)$ in $H$. For all $0\leq l \leq p-1$ and all $0\leq i\leq d^l$ we include the edges $(u_i^l,u_{i\cdot d+1}^{l+1})$ in $H$, and if $i$ is odd, we also include the path connecting $u_i^l$ to $u_{i\cdot d^{p-l}+1}$ in $H$. An example is given in Figure \ref{fig:ham5}.
\begin{figure}[htbp]
\begin{center}
\includegraphics[width=0.7\textwidth]{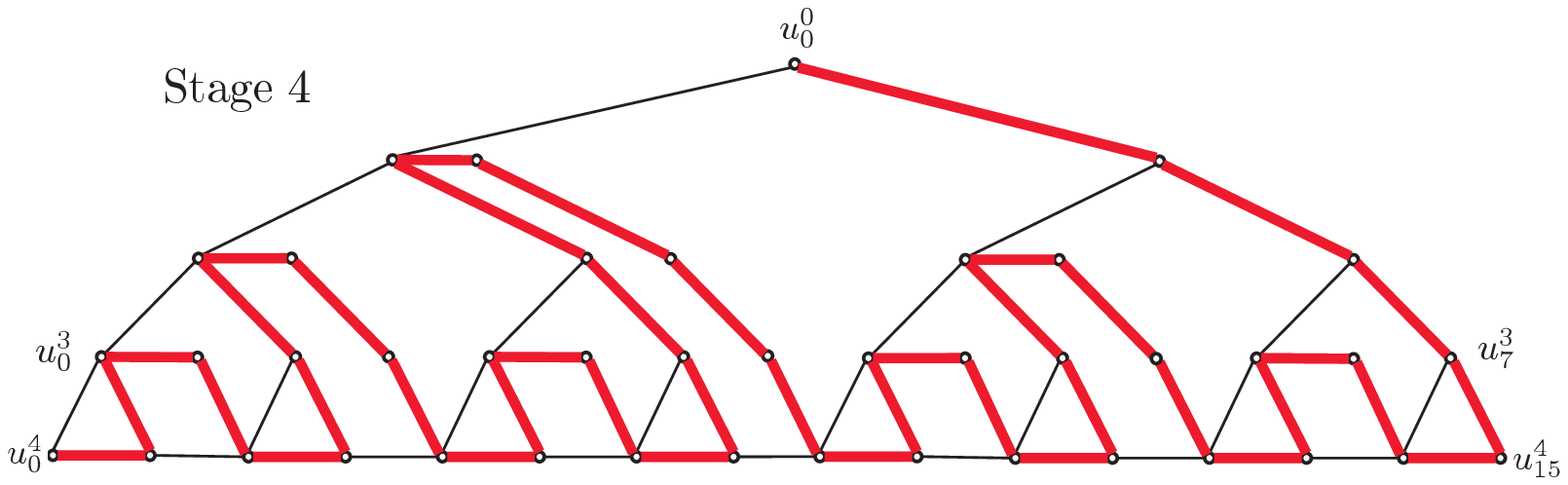}
\end{center}
\caption{The modified tree for $p=4$ and $d=2$. The red bold edges form path $S_3$.}\label{fig:ham5}
\end{figure}

\textbf{Stage 5} We now connect end points of the three paths.
Let us investigate the six endpoints of the three paths.
\begin{itemize}
\item End points of $S_3$ are $u_0^0$ and $u_0^p$ (see Figure~\ref{fig:ham5}).
\item Path $S_2$ has both endpoints on the $r$-side (see Figure~\ref{fig:ham2}). Let us denote these endpoints by $e_1$ and $e_2$. Depending on the input strings, endpoint $e_1$ is either $v_{m-1}^{3}$ or $v_{m}^{2}$ and the other endpoint $e_2$ is either $v_{m-1}^{\Gamma-2}$ or $v_{m-1}^{\Gamma-3}$.
\item The endpoints of $S_1$ are both on the $r$-side (see Figure~\ref{fig:ham1}): $v_{m}^{\Gamma}$ and $v_{m}^1$.
\end{itemize}
Now we connect those endpoints in the following way.
\begin{itemize}
\item Connect $S_1$ and $S_2$, each at one endpoint on the $r$-side, by letting edge $(e_1,v_{m}^1)$ participate in $H$.
\item We connect the endpoint $u_0^0$ of $S_3$  to the endpoint $v_{m}^{\Gamma}$ of $S_1$ by marking an edge $(u_0^0, v_{m}^{\Gamma})$ as participating.
\item Connect the endpoint $e_2$ to $v$ and $v$ to the endpoint $u_0^p$ of $S_3$ by including the corresponding edges of $G(\Gamma,2,p)'$ in $H$.
\end{itemize}
An example is given in Figure \ref{fig:ham2}.

If the strings are equal then the result is a Hamiltonian cycle since it contains the paths $S_1$, $S_2$, and $S_3$ that will be three disjoint paths (connected to be a cycle) that cover all nodes.
Now we prove that if the strings are not equal, $H$ will not be a Hamiltonian cycle. Let $i$ be a position in which $\indVarS$ and $\indVarR$ differ. Let us consider the case that $x_i=0$ and $y_i=1$. Then the sequence $x'_{1+6i},\dots,x'_{6+6i}$ will be $100100$ while the sequence $y'_{1+6i},\dots,y'_{6+6i}$ will be $110110$. When we look at the part of the graph $H$ corresponding to this sequence (see Figure \ref{fig:ham3}), we see that $H$ can not be a cycle and thus not a Hamiltonian cycle: due to $y'_{2+6i}=0$ and $x'_{2+6i}=1$ there are no edges on the $s$- nor $r$-side of level $2+6i$ connecting the part of $S_1$ below level $2+6i$ to the part of $S_1$ above level $2+6i$. There will also be no edges of $S_2$ that accidentally connect those two parts to each other.
The case that $x_i=1$ and $y_i=0$ is treated the same way: due to the construction of $x'$ and $y'$ the sequence $x'_{6b+1+6i},\dots,x'_{6b+6+6i}$ is $100100$ and $y'_{6b+1+6i},\dots,y'_{6b+6+6i}$ will be $110110$ and we can use exactly the same argument as before.
\begin{figure}
\begin{center}
\includegraphics[width=0.5\textwidth]{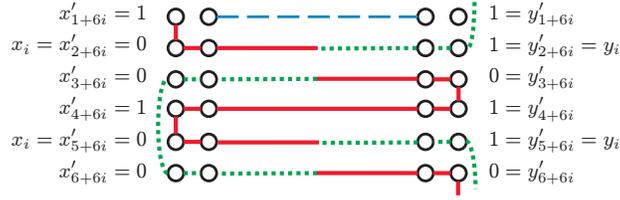}
\end{center}
\caption{Example of the case that $x_i=0$ and $y_i=1$.}\label{fig:ham3}
\end{figure}

Now consider an algorithm $\cA_{ham}$ for the Hamiltonian cycle verification problem. When $\cA_{ham}$ terminates, vertex $s$
determines its output for the equality problem by stating that both input strings are equal if and only if $\cA_{ham}$ verified that $H$ is a Hamiltonian cycle.

Hence a fast algorithm for the Hamiltonian cycle problem on $G(\Gamma,2,p)'$ can be used to correctly solve the given instance of the equality problem on $G(\Gamma,2,p)'$ faster. This contradicts the lower bound of the equality verification problem, which holds for all $d$ (here we used $d=2$).

%\end{proof}

%
%Combined with Corollary~\ref{cor:eq}, we now have
%\begin{theorem}\label{thm:ham}
%For every $\Gamma$ and corresponding $p$, any distributed algorithm for solving Hamiltonian cycle problem on the graph $(\Gamma,2,p)'$ in the
%$B$ model for $B\ge 3$ requires $\Omega(\min(d^p,\Gamma/dpB))$ time.
%\end{theorem}
%Using the same analysis as in Section~\ref{sec:randomized_lb}, we obtain:
%\begin{corollary}
%\label{cl:ham}
%For any $p\geq 1$, $B\geq 3$, and $n\in \{2^{2p+1}pB, 3^{2p+1}pB, \ldots\}$, any distributed algorithm for the Hamiltonian cycle problem in the $B$ model
%for $B\ge 3$ requires $\Omega(\frac{n}{pB}^{\frac{1}{2}-\frac{1}{2(2p+1)}})$ time on
%some $n$-vertex graph of diameter $2p+2$.
%\end{corollary}

\subsection{Lower Bound of Spanning Tree and Path Verification Problems}\label{sec:other_deterministic} The remaining two deterministic lower bounds follow from the lower bound of the Hamiltonian cycle verification problem, as follows.

%\begin{lemma}\label{lem:spanningtree}
%For any $p\geq 1$, $B\geq 3$, and $n\in \{2^{2p+1}pB, 3^{2p+1}pB, \ldots\}$, any $B\geq 3$, any deterministic algorithm for spanning tree verification requires $\Omega(\frac{n}{pB}^{\frac{1}{2}-\frac{1}{2(2p+1)}})$ time in some family of $n$-vertex graph of diameter $2p+2$.
%\end{lemma}
%\begin{proof}
\paragraph{Spanning Tree Verification Problem} We reduce Hamiltonian cycle verification to spanning tree verification using $O(D)$ rounds using the following observation: $H$ is a Hamiltonian cycle if and only if every vertex has degree exactly two and $H\setminus e$, for any edge $e$ in $H$, is a spanning tree.

Therefore, to verify that $H$ is a Hamiltonian cycle, we first check whether every vertex has degree exactly two in $H$. If this is not true then $H$ is not a Hamiltonian cycle. This part needs $O(D)$ rounds. Next, we check if $H\setminus \{e\}$, for any edge $e$ in $H$, is a spanning tree. We announce that $H$ is a Hamiltonian cycle if and only if $H\setminus \{e\}$ is a spanning tree.
%\end{proof}

%\begin{lemma}\label{lem:path}
%For any $p\geq 1$, $B\geq 3$, and $n\in \{2^{2p+1}pB, 3^{2p+1}pB, \ldots\}$, any $B\geq 3$, any deterministic algorithm for path verification requires $\Omega((\frac{n}{pB}^{\frac{1}{2}-\frac{1}{2(2p+1)}}))$ time in some family of $n$-vertex graph of diameter $2p+2$.
%\end{lemma}
%\begin{proof}
\paragraph{Simple Path Verification Problem} Similar to the above proof, we reduce Hamiltonian cycle verification to path verification using $O(D)$ rounds using the following observation: $H$ is a Hamiltonian cycle if and only if every vertex has degree exactly two and $H\setminus e$ is a path (without cycles).
%\end{proof}

\section{Hardness of Distributed Approximation}\label{sec:approxalgo}
In this section we show a time lower bound of $\Omega(\sqrt{n/(B\log n)})$ for Monte Carlo randomized approximation algorithms of many problems (defined in Section~\ref{subsec:approx_def}), as in the following theorem.

%For distributed approximation problems such as MST, we assume that a weight function $\omega: E \rightarrow \mathbb{R}^+$ associated with the graph assigns a nonnegative real weight $\omega(e)$ to each edge
%$e=(u,v)\in E$. Initially, the weight $\omega(e)$ is known only to the adjacent
%vertices, $u$ and $v$.
%We assume that the edge weights are bounded by a polynomial in $n$
%(the number of vertices). It is assumed that $B$ is large enough
%to allow the transmission of any edge weight in a single message.

%We show the hardness of distributed approximation for many problems, as in the theorem below.
%
%

\begin{theorem}\label{thm:randomized_approximation}
For any polynomial function $\alpha(n)$, numbers $p$, $B\geq 1$, and $n\in \{2^{2p+1}pB, 3^{2p+1}pB, \ldots\}$, there exists a constant $\epsilon>0$  such that any $\alpha(n)$-approximation $\epsilon$-error distributed algorithm for any of the following problems requires $\Omega((\frac{n}{pB})^{\frac{1}{2}-\frac{1}{2(2p+1)}})$ time on some $\Theta(n)$-vertex graph of diameter $2p+2$ in the $B$ model: minimum spanning tree~\cite{Elkin06,PelegR00}, shallow-light tree~\cite{peleg}, $s$-source distance~\cite{Elkin05}, shortest path tree~\cite{Elkin06},  minimum routing cost spanning tree~\cite{WuLBCRT99}, minimum cut~\cite{Elkin-sigact04}, minimum $s$-$t$ cut, shortest $s$-$t$ path and generalized Steiner forest~\cite{KhanKMPT08}.
\end{theorem}

%\begin{theorem}\label{thm:randomized_approximation}
%For any function $\alpha(n)$ and $n\geq 2$, there exists a constant $\epsilon>0$ such that any $\alpha(n)$-approximation $\epsilon$-error distributed algorithm for any of the following problems requires $\Omega(\sqrt{n}/(B\log n))$ time on some $n$-vertex graph of diameter $O(\log n)$ in the $B$ model where $B\geq 3$: minimum spanning tree~\cite{Elkin06,PelegR00}, shortest $s$-$t$ path, $s$-source distance~\cite{Elkin05}, $s$-source shortest path tree~\cite{Elkin06}, minimum cut~\cite{Elkin-sigact04}, minimum $s$-$t$ cut, maximum cut, minimum routing cost spanning tree\danupon{Cite what?}, shallow-light tree~\cite{peleg}, least-element list~\cite{KhanKMPT08}, and generalized Steiner forest~\cite{KhanKMPT08}.
%\end{theorem}

%We note that this is a slight relaxation of the minimum spanning tree problem usually considered in literature since we do not need to find the tree itself.

In particular, for graphs with diameter $D=4$, we get $\Omega((n/B)^{1/3})$ lower bound and for graphs with diameter $D=\log{n}$ we get $\Omega(\sqrt{n/(B\log n)})$. Similar analysis also leads to a $\Omega(\sqrt{n/B})$ lower bound for graphs of diameter $n^\delta$ for any $\delta>0$, and $\Omega((n/B)^{1/4})$ lower bound for graphs of diameter three using the same analysis as in \cite{Elkin06}.

The main tool used in this section is the randomized lower bound of network verification problems defined in Section~\ref{subsec:verification_def} and proved in Section~\ref{sec:randomized_lb}.

The main proof idea is similar to the proof that the Traveling Salesman Problem on general graphs cannot be approximated within $\alpha(n)$ for any polynomial computable function $\alpha(n)$ (see, e.g., \cite{Vazirani-book}): We will define a weighted graph $G'$ in such a way that if the subgraph $H$ satisfies the desired property then the approximation algorithm must return some value that is at most $f(n)$, for some function $f$. Conversely, if $H$ does not satisfy the property, the approximation algorithm will output some value that is strictly more than $f(n)$. Thus, we can distinguish between the two cases.

To highlight the main idea, we first prove the theorem for the minimum spanning tree problem in the next subsection. Proofs of other problems are in Subsection~\ref{subsec:approxalgo_FULL}.

\subsection{Lower Bound of Approximating the Minimum Spanning Tree}\label{subsec:approx_MST}
Recall that in the minimum spanning tree problem, we are given a connected graph $G$ and we want to compute a minimum spanning tree (i.e., a spanning tree of minimum weight). At the end of the process each vertex knows which edges incident to it are in the output tree.

%\paragraph{Approximation algorithms:}
%Recall the following standard notions of an {\em approximation algorithm}. We say that a randomized algorithm $\mathcal A$ is {\em $\alpha$-approximation $\epsilon$-error} if, for any input instance $\mathcal I$, algorithm $\mathcal A$ outputs a solution that is at most $\alpha$ times the optimal solution of $\mathcal I$ with probability at least $1-\epsilon$. Therefore, in the minimum spanning tree, an $\alpha$-approximation $\epsilon$-error algorithm should output a number that is at most $\alpha$ times the total weight of the minimum spanning tree, with probability at least $1-\epsilon$.

%\begin{proof}[Proof of Theorem~\ref{thm:randomized_approximation}] (This proof is only for the case of minimum spanning tree.)
Let $\mathcal{A}_\epsilon$ be an $\alpha(n)$-approximation $\epsilon$-error algorithm for the minimum spanning tree problem. We show that
%we can use the algorithm
$\mathcal{A}_\epsilon$ can be used to solve the spanning connected subgraph verification problem using the same running time (thus the lower bound proved in Theorem~\ref{thm:all_verification} applies).

%The proof now is straightforward.
%To verify if a subgraph $H$ is a spanning connected subgraph, we
To do this, construct a weight function on edges in $G$, denoted by $\omega$, by assigning weight $1$ to all edges in $H$ and weight $n\alpha(n)$ to all other edges. Note that constructing $\omega$ does not need any communication since each vertex knows which edges incident to it are members of $H$. Call this weighted graph $G'$. Now we find the weight $W$ of the $\alpha(n)$-approximated minimum spanning tree of $G'$ using $\mathcal{A}_\epsilon$ and announce that $H$ is a spanning connected subgraph if and only if $W$ is less than $n\alpha(n)$.

%
%We now show the correctness of the above algorithm. In particular,
%
%We show
We will show that the weighted graph $G'$ has a spanning tree of weight less than $n\alpha(n)$ if and only if $H$ is a spanning connected subgraph of $G$ and thus the algorithm above is correct. Suppose that $H$ is a spanning connected subgraph. Then, there is a spanning tree that is a subgraph of $H$ and has weight $n-1<n\alpha(n)$ since $\alpha(n)\geq 1$. Thus the minimum spanning tree has weight less than $n\alpha(n)$.
Conversely, suppose that $H$ is not a spanning connected subgraph. Then, any spanning tree of $G'$ must contain an edge not in $H$. Therefore, any spanning tree has weight at least $n\alpha(n)$ as claimed.

%
%Our MST lower bound here matches the lower bound of exact MST algorithms and improves the lower bound of $\Omega(\sqrt{\frac{n}{\alpha B}})$ by Elkin~\cite{Elkin06}. %Our lower bound for $s$-source distance complements the results in \cite{Elkin05}.

\subsection{Lower Bounds of Other Problems}\label{subsec:approxalgo_FULL}
%\begin{proof}[Proof of Theorem~\ref{thm:randomized_approximation}]
%
%We have already used this technique to proof the lower bound of the MST problem in Section~\ref{sec:approxalgo}. We now show lower bounds of other problems using the same technique.
%
%
%For the minimum spanning tree problem, we do the following. \stephan{Istn't this proof for MST already in the main-part?}Let $\mathcal{A}$ be an $\alpha(n)$-approximation algorithm for the minimum spanning tree problem. We show that we can use this algorithm $\mathcal A$ to solve the spanning connected subgraph verification problem using the same running time.
%%
%To verify if a subgraph $H$ is a spanning connected subgraph, first we construct a weighted graph $G'$ by assigning weight $1$ to all edges in $H$ and $n\alpha(n)$ to all other edges. Suppose that $H$ is a spanning connected subgraph. Then, the minimum spanning tree of $G'$ has weight $n-1$ and thus $\mathcal{A}$ must return a tree with weight at most $(n-1)\alpha(n)$. Conversely, if $H$ is not a spanning connected subgraph then the minimum spanning tree will contain at least one edge of weight $n\alpha(n)$. Therefore, $H$ is a spanning connected subgraph if and only if the weight of the minimum spanning tree returned by $\mathcal{A}$ is at most $(n-1)\alpha(n)$.
%%
%The lower bound for verifying spanning connected subgraph (cf. Theorem~\ref{thm:all_verification}) thus applies to $\mathcal{A}$.
%
We now prove the remaining lower bounds.

\paragraph{Shallow-light tree problem} The lower bound for the shallow-light tree problem follows immediately from the lower bound of the MST problem when we set the length of every edge to be one and radius requirement to be $n$. In this case, the spanning tree satisfies the radius requirement and so the minimum-weight shallow-light tree becomes the minimum spanning tree.

\paragraph{$s$-source distance and shortest path tree problems} We construct the graph $G'$ as in Subsection~\ref{subsec:approx_MST} and the lower bounds follow in a similar way: $H$ is a spanning connected subgraph if and only if the distance from $s$ to every node is at most $n-1$ (i.e., $\mathcal{A}$ has approximated the distance to be at most $(n-1)\alpha(n)$), which is true if and only if the shortest path spanning tree contains only edges of weight one (i.e., the total weight of the shortest path spanning tree is at most $(n-1)\alpha(n)$).

\paragraph{Minimum routing cost spanning tree problem} We modify the construction of $G'$ in Subsection~\ref{subsec:approx_MST} as follows. We assign weight one to edges in $H$ and $n^3\alpha(n)$ to other edges. Observe that if $H$ is a spanning connected subgraph, the routing cost between any pair will be at most $n-1$ and thus the cost of the $\alpha(n)$-approximation minimum routing cost spanning tree will be at most $(n-1){n \choose 2}\alpha(n)<n^3\alpha(n)$. Conversely, if $H$ is not a spanning connected subgraph, some pair of nodes will have routing cost at least $n^3\alpha(n)$ and thus the minimum routing cost spanning tree will cost at least $n^3\alpha(n)$.

\paragraph{Minimum cut problem} We define $\bar{G'}$ by assigning weight one to all edges in $\bar{H}=(V, E(G)\setminus E(H))$ and $n\alpha(n)$ to all other edges and use the fact that $\bar{H}$ is a cut if and only if $\bar{G'}$ has a minimum cut of weight at most $n-1$, i.e., $\mathcal{A}$ outputs a value of at most $(n-1)\alpha(n)$.

\paragraph{Minimum $s$-$t$ cut problem} The reduction is the same as in the case of the minimum cut problem. Observe that $s$ and $t$ are {\em not} connected in $H$ if and only if $\bar{H}$ is a $s$-$t$ cut which in turn is the case if and only if $\bar{G'}$ has a minimum $s-t$ cut of weight $n-1$. Thus, the lower bound of $s$-$t$ cut verification problem implies the lower bound of this problem.

\paragraph{Shortest $s$-$t$ path problem} We again construct $G'$ as in Subsection~\ref{subsec:approx_MST}. Observe that $s$ and $t$ are in the same connected component in $H$ if and only if  the distance between $s$ to $t$ in $G'$ is at most $n-1$, i.e., $\mathcal{A}$ outputs a value of at most $(n-1)\alpha(n)$. The lower bound follows from the lower bound of $s$-$t$ connectivity verification problem.

\paragraph{Generalized Steiner forest problem} We will reduce from the lower bound of $s$-$t$ connectivity. We have only one set $V_1=\{s, r\}$. Construct $G'$  as in Subsection~\ref{subsec:approx_MST}. Observe that the minimum generalized Steiner forest will have weight at most $n-1$ if $H$ is $s$-$t$ connected and weight at least $n\alpha(n)$ otherwise. (Recall that $G$ is assumed to be connected in the problem definition.)

\section{Tightness of Lower Bounds of Verification Problems}\label{sec:tightness}
We note that almost all lower bounds of verification problems stated so far are almost tight. To show this we will present deterministic $O(\sqrt{n}\log^*n+D)$-time algorithms for all verification problems except the least-element list verification problem. This upper bound almost matches the $\tilde \Omega(\sqrt{n})$ lower bounds shown in previous sections.
%
%the $s$-$t$ connectivity, $k$-component, connectivity, cut, $s$-$t$-cut, bipartiteness, edge on all path, and simple path verification problems. Algorithms for all other problems stated in this paper can be found using the reductions given in Figure \ref{fig:all_reductions}.
%
Our main tool is the MST algorithm by Kutten and Peleg~\cite{KuttenP98} and the connected component algorithm by Thurimella~\cite[Algorithm~5]{Thurimella97}.

Recall that in the MST problem, we are given a weighted network $G$ (that is the weight of each edge is known to the nodes incident to it) and we want to find a minimum spanning tree (for each edge $e$, nodes incident to it know whether $e$ is in the MST or not.) Kutten and Peleg~\cite{KuttenP98} showed that this problem can be solved by a $O(\sqrt{n}\log^*n+D)$-time distributed deterministic algorithm.

We also note the following connected component algorithm by Thurimella \cite{Thurimella97}. Given a subgraph $H$ of $G$, the algorithm outputs a label $\ell(v)$ for each node $v$ such that for any two nodes $u$ and $v$, $\ell(u)=\ell(v)$ if and only if $u$ and $v$ are in the same connected component.  Theorem~6 in \cite{Thurimella97} states that the distributed time complexity of this algorithm is $O(D+f(n)+g(n)+\sqrt{n})$ where $f(n)$ and $g(n)$ are the distributed time complexities of finding a MST and a $\sqrt{n}$ -dominating set, respectively. Due to \cite{KuttenP98} we have that $f(n)=g(n)=O(D+\sqrt{n}\log^* n)$.

We are now ready to present algorithms for our verification problems.

%\paragraph{Deterministic algorithms almost matching the deterministic lower bounds:}We need to give upper bounds for the $k$-spanning tree and path %and $s$-$t$-connectivity in degree-$2$ graphs verification problems.

\paragraph{Spanning connected subgraph, spanning tree, cycle containment and connectivity verification problems} We construct a weighted graph $G'$ by assigning weight zero to all edges in $H$ and weight one to other edges (for each edge $e$, nodes incident to it know its weight). Observe the followings.
\begin{itemize}
\item $H$ is a spanning tree if and only if the MST of $G'$ has cost zero.
\item $H$ is a spanning connected subgraph if and only if the MST of $G'$ has cost zero.
\item $H$ contains no cycle if and only if all edges in $H$ are in the MST of $G'$, i.e., the cost of the MST of $G'$ is $n-1-|E(H)|$ where $|E(H)|$ is the number of edges in $H$.
\item $H$ is connected if and only if there are $|V(H)|-1$ edges in the MST that have cost zero, where $V(H)$ is the set of nodes incident to some edges in $H$. This is because all edges in a spanning forest of $H$ can be used in the MST and there are less than $V(H)-1$ such edges if and only if $H$ is not connected.
\end{itemize}

Thus, we can verify these properties of $H$ by finding a minimum spanning tree of $G'$ using the $O(\sqrt{n}\log^*n+D)$-time algorithm of Kutten and Peleg~\cite{KuttenP98}.

\paragraph{Cut verification problem}
To verify if $H$ is a cut, we simply verify if $G$ after removing the edges in $H$, i.e. $H'=(V, E(G)\setminus E(H))$, is connected.

%Now we create $H^{k-1}:=H^k\setminus T^k$. $H^{k-1}$ is a $k-1$-spanning tree if and only if $H^k$ is a $k$-spanning tree. If all $T^j$ were spanning trees, after $k$ iterations we are left with $H^0$ which contains no nodes and no edges if and only if $H^k$ was a $k$-spanning tree.

%\paragraph{Connectivity verification problem} We construct $G'$ by putting weight $1$ on edges in $H$ and $2$ on other edges and find the MST using an algorithm in \cite{KuttenP98}. Observe that $H$ has at most $k$ connected component if and only if there are at most $k-1$ edges of weight $2$ in the MST, i.e., the MST has weight at most $n-1+(k-1)$.

\paragraph{$s$-$t$ connectivity verification problem} We simply run Thurimella's algorithm (as explained above) and verify whether $s$ and $t$ are in the same connected component by verifying whether $\ell(s)=\ell(t)$.

\paragraph{Edge on all paths verification problem} Observe that $e$ lies on all paths between $u$ and $v$ if and only if $u$ and $v$ are disconnected in $H\setminus \{e\}$. Thus, we can use the $s$-$t$ connectivity verification algorithm above to check this.

\paragraph{$s$-$t$ cut verification problem}
To verify if $H$ is a $s$-$t$ cut, we simply verify $s$-$t$ connectivity of $G$ after removing the edges in $H$ (i.e.,  $H'=(V, E(G)\setminus E(H))$).

\paragraph{$e$-cycle verification problem}
To verify if $e$ is in some cycle of $H$, we simply verify $s$-$t$ connectivity of $H'=H\setminus \{e\}$ where $s$ and $t$ are the end nodes of $e$.

\paragraph{Bipartiteness verification problem} We run Thurimella's algorithm to find the connected components of $H$. We note that this algorithm can in fact output a rooted spanning tree of each connected component of $H$ and make each node knows its level in such a tree. This level implies a natural two-coloring of nodes in $H$. Now all nodes check if their neighbors have a color different from their own color. They will have a different color if and only if $H$ is bipartite.

\section{Conclusion}\label{sec:conclusion}

We initiate the systematic study of verification problems in the context of distributed network algorithms and present a uniform lower bound for several  problems.  We also show how these verification bounds can be used to obtain lower bounds on exact and approximation algorithms for many problems. Our techniques exploit well-known
bounds in communication complexity to show lower bounds in distributed computing. Our techniques give
a general and powerful methodology  for showing non-trivial lower bounds for various problems in distributed computing.

Several problems remain open. A general direction for extending all of this work is to study similar verification problems in special classes of graphs, e.g., a complete graph. A few specific open questions include proving better lower or upper bounds for the problems of shortest $s$-$t$ path, single-source distance computation, shortest path tree, $s$-$t$ cut, minimum cut. (Some of these problems were also asked in \cite{Elkin-sigact04}.) Also, showing randomized bounds for Hamiltonian path, spanning tree, and simple path verification remains open.
 % Conclusion + Open problems

%\section*{Acknowledgments}
%The author thanks the anonymous authors whose work largely
%constitutes this sample file. He also thanks the INFO-TeX mailing
%list for the valuable indirect assistance he received.

\bibliographystyle{siam}
\bibliography{distributed-verification}

%\newpage
%\appendix
%\section*{Appendix}
%\input{other_lb} % Randomized lower bounds of other verification problems that we do not include in the main text
%\input{deterministic_lb} % Deterministic lower bound of spanning tree verification
%\input{approx_FULL} % Full details of lower bounds of approximation algorithms
%

\end{document}